\documentclass[12pt,reqno]{amsart}

\usepackage{amsmath,amsthm,amssymb,comment,fullpage}
\usepackage{times}
\usepackage[T1]{fontenc}
\usepackage{mathrsfs}
\usepackage{latexsym}
\usepackage[dvips]{graphics}
\usepackage{epsfig}
\usepackage{accents}
\usepackage{float}
\usepackage{amsmath,amsfonts,amsthm,amssymb,amscd}
\input amssym.def
\input amssym.tex
\usepackage{color}
\usepackage{hyperref}
\usepackage{url}
\usepackage{breakurl}
\usepackage{comment}
\usepackage{mathtools}
\newcommand{\bburl}[1]{\textcolor{blue}{\url{#1}}}

\newcommand{\E}{\mathbb{E}}

\renewcommand{\E}{\mathbb{E}}

\newcommand{\complexI}{\hat{i}}
\newcommand{\quaternionI}{\hat{i}}
\newcommand{\quaternionJ}{\hat{j}}
\newcommand{\quaternionK}{\hat{k}}

\numberwithin{equation}{section}

\newtheorem{thm}{Theorem}[section]

\newtheorem{lem}[thm]{Lemma}
\newtheorem{prop}[thm]{Proposition}
\newtheorem{exa}[thm]{Example}
\newtheorem{defi}[thm]{Definition}

\theoremstyle{plain}

\newtheorem{lemma}[thm]{Lemma}

\newtheorem{theorem}[thm]{Theorem}

\newtheorem{rem}[thm]{Remark}

\newtheorem{rek}[thm]{Remark}

\newcommand\be{\begin{equation}}
\newcommand\ee{\end{equation}}
\newcommand\bea{\begin{eqnarray}}
\newcommand\eea{\end{eqnarray}}
\newcommand\bi{\begin{itemize}}
	\newcommand\ei{\end{itemize}}
\newcommand\ben{\begin{enumerate}}
	\newcommand\een{\end{enumerate}}
\newcommand\bc{\begin{center}}
	\newcommand\ec{\end{center}}
\newcommand\ba{\begin{array}}
	\newcommand\ea{\end{array}}

\newcommand{\gl}{\lambda}

\newcommand{\R}{\ensuremath{\mathbb{R}}}
\newcommand{\C}{\ensuremath{\mathbb{C}}}
\newcommand{\Z}{\ensuremath{\mathbb{Z}}}

\newcommand{\N}{\mathbb{N}}

\newcommand{\hr}[1]{\href{#1}{\url{#1}}}

\renewcommand \l {\lambda}

\renewcommand{\H}{\mathbb{H}}

\newcommand{\eps}{\epsilon}

\renewcommand{\Re}[1]{\text{Re}(#1)}
\renewcommand{\Re}{\mathfrak{Re}}

\newcommand{\var}{\text{Var}}
\newcommand{\pfrac}[2]{\left(\frac{#1}{#2}\right)}

\DeclareMathOperator{\tr}{Tr}
\newcommand{\etr}{\mathbb{E}\;\tr\;}
\newcommand{\ektr}{\mathbb{E}_k\;\tr\;}
\newcommand{\ecktr}{\mathbb{E}_k^\C\;\tr\;}
\newcommand{\ehktr}{\mathbb{E}_k^\H\;\tr\;}
\newcommand{\Eo}[1]{\underaccent{#1}{\mathbb{E}}}
\newcommand{\kmat}{b}

\newcommand*{\reff}[1]{\hyperref[#1]{\ref{#1}}}
\newcommand*{\pez}[1]{\left( #1\right)}
\newcommand*{\on}{\operatorname}
\DeclareMathOperator{\rank}{rank}

\DeclarePairedDelimiter\abs{\lvert}{\rvert}%
\DeclareMathOperator{\prob}{Pr}

\title{Random Matrix Ensembles with Split Limiting Behavior}

\author{Paula Burkhardt}
\email{\textcolor{blue}{\href{peb02012@pomona.edu}{peb02012@pomona.edu}}}
\address{Department of Mathematics, Pomona College, Claremont, CA 91711}

\author{Peter Cohen}
\email{\textcolor{blue}{\href{mailto:pcohen@bowdoin.edu, petercohen33@gmail.com}{pcohen@bowdoin.edu, petercohen33@gmail.com}}}
\address{Department of Mathematics, Bowdoin College, Brunswick, ME 04011}

\author{Jonathan DeWitt}
\email{\textcolor{blue}{\href{mailto:jon.dewitt@gmail.com}{jon.dewitt@gmail.com}}}
\address{Department of Mathematics and Statistics, Haverford College, Haverford, PA 19041}

\author{Max Hlavacek}
\email{\textcolor{blue}{\href{mailto:mhlavacek@g.hmc.edu}{mhlavacek@g.hmc.edu}}}
\address{Department of Mathematics, Harvey Mudd College, Claremont, CA 91711}

\author{Steven J. Miller}
\email{\textcolor{blue}{\href{mailto:sjm1@williams.edu, Steven.Miller.MC.96@aya.yale.edu}{sjm1@williams.edu,Steven.Miller.MC.96@aya.yale.edu}}}
\address{Department of Mathematics and Statistics, Williams College, Williamstown, MA 01267}

\author{Carsten Sprunger}
\email{\textcolor{blue}{\href{mailto:csprun@umich.edu}{csprun@umich.edu}}}
\address{Department of Mathematics, University of Michigan, Ann Arbor, MI 48109}

\author{Yen Nhi Truong Vu}
\email{\textcolor{blue}{\href{mailto:ytruongvu17@amherst.edu}{ytruongvu17@amherst.edu}}}
\address{Department of Mathematics and Statistics, Amherst College, Amherst, MA 01002}

\author{Roger Van Peski}
\email{\textcolor{blue}{\href{mailto:rpeski@princeton.edu}{rpeski@princeton.edu}}}
\address{Department of Mathematics, Princeton University, Princeton, NJ 08544}

\author{Kevin Yang}
\email{\textcolor{blue}{\href{mailto:kevinyang@college.harvard.edu}{kevinyang@college.harvard.edu}}}
\address{Department of Mathematics, Harvard University, Cambridge, MA 02138}

\thanks{The authors were partially supported by NSF Grants DMS1265673, DMS1561945, DMS1449679 and DMS1347804, Amherst College, the University of Michigan, Princeton University and Williams College. We thank  Eyvindur Palsson, Arup Bose and Aaditya Sharma for helpful conversations. This work was supervised by the fifth named author at the Williams SMALL REU program; the first, third and ninth authors determined the behavior in the bulk and began the investigations of the blip during the 2015 SMALL program, which were undertaken by the remaining authors the following year.}



\date{\today}

\begin{document}

\begin{abstract} We introduce a new family of $N\times N$ random real symmetric matrix ensembles, the $k$-checkerboard matrices, whose limiting spectral measure has two components which can be determined explicitly. All but $k$ eigenvalues are in the bulk, and their behavior, appropriately normalized, converges to the semi-circle as $N\to\infty$; the remaining $k$ are tightly constrained near $N/k$ and their distribution converges to the $k \times k$ hollow GOE ensemble (this is the density arising by modifying the GOE ensemble by forcing all entries on the main diagonal to be zero). Similar results hold for complex and quaternionic analogues. We isolate the two regimes by using matrix perturbation results and a nonstandard weight function for the eigenvalues, then derive their limiting distributions using a modification of the method of moments and analysis of the resulting combinatorics.
\end{abstract}

\maketitle
\tableofcontents


\section{Introduction}

\subsection{Background}\label{sec:background}

Since their introduction by Wishart \cite{Wis} in the 1920s in statistics, the distribution of eigenvalues of random matrix ensembles have played a major role in a variety of fields, especially in nuclear physics and number theory; see for example the surveys \cite{Bai, BFMT-B, Con, FM, KaSa, KeSn} and the textbooks \cite{Fo, Meh, MT-B, Tao2}. One of the central results in the subject is Wigner's semi-circle law. Inspired by studies of energy levels of heavy nuclei, Wigner conjectured that their energy levels are well-modeled by eigenvalues of a random matrix ensemble, and he and others proved that in many matrix ensembles the distribution of the scaled eigenvalues of a typical matrix converge, in some sense, to the semi-circle distribution \cite{Wig1, Wig2, Wig3, Wig4, Wig5}.

Which matrix ensemble models the system depends on its physical symmetries. Though the most used in physics and number theory are the Gaussian Orthogonal, Unitary and Symplectic Ensembles, it is of interest to study other families. In many cases the additional symmetry constraints on the  matrix (for example, requiring it to be Toeplitz or circulant or arising from a $d$-regular graph) lead to a different density of states. There is now an extensive literature on the density of eigenvalues of special ensembles; see for example \cite{Bai, BasBo1, BasBo2, BanBo, BLMST, BCG, BHS1, BHS2, BM, BDJ, GKMN, HM, JMRR, JMP, Kar, KKMSX, LW, MMS, MNS, MSTW, McK, Me, Sch}, where many of them have limiting spectral measures different than the semi-circle (though recent work, see \cite{ERSY, ESY, TV1, TV2} among others, shows that in many cases the spacing between normalized eigenvalues is universal and equals that of the Gaussian ensembles).

In many of these special ensembles while one is able to prove the density of eigenvalues of a typical matrix converges to a limiting spectral measure, one cannot write down a nice, closed-form expression for this limiting distribution (notable exceptions are $d$-regular graphs \cite{McK}, block circulant matrices \cite{KKMSX} and palindromic Toeplitz matrices \cite{MMS}). 
In what follows we study a new ensemble of `checkerboard' matrices, the eigenvalues of which are split into two types, each of which converges to a different limiting spectral distribution which can be solved for explicitly. Most of the eigenvalues are of order $\sqrt{N}$ and converge to a semi-circle; however, a small number are of size $\Theta(N)$ and converge to new limiting measures related to the Gaussian ensembles. We define these matrices in the next section, and then summarize our findings and the techniques developed to study such split behavior.

\subsection{Generalized Checkerboard Ensembles}

Our arguments apply with only minor modification to the reals, complex numbers and quaternions, and show connections between the checkerboard and Gaussian ensembles. As we often use $i$ as an index of summation, we use $\complexI:=\sqrt{-1}$ and similarly $\quaternionI, \quaternionJ$ and $\quaternionK$ for the quaternions. Additionally, we index the entries $m_{ij}$ of a matrix starting at $0$ to simplify certain congruence conditions.

\begin{defi}\label{defn:checkerboard}
Fix $D =\R, \C$ or $\H$, $k \in \N$, $w \in \R$. Then the $N \times N$ $(k,w)$-checkerboard ensemble over $D$ is the ensemble of matrices $M = (m_{ij})$ given by
\begin{equation}\label{def $k$-checkerboard matrix}
m_{ij}\ =\ \begin{cases}
a_{ij} & \text{{\rm if} } i \not\equiv j \bmod k\\
w & \text{{\rm if} } i \equiv j \bmod k
\end{cases}
\end{equation}
where $a_{ij}=\overline{a_{ji}}$ and
\begin{equation}
a_{ij}\ =\ \begin{cases}
r_{ij} &\mbox{{\rm if} } D=\R \\
\frac{r_{ij}+b_{ij}\complexI}{\sqrt{2}} &\mbox{{\rm if} } D=\C \\
\frac{r_{ij}+b_{ij}\quaternionI + c_{ij}\quaternionJ + d_{ij}\quaternionK}{2} &\mbox{{\rm if} } D=\H
\end{cases}
\end{equation}
with $r_{ij}$, $b_{ij}$, $c_{ij}$, and $d_{ij}$ i.i.d. random variables with mean 0, variance 1, and finite higher moments, and the probability measure on the ensemble given by the natural product probability measure. We refer to the $(k,1)$-checkerboard ensemble over $D$ simply as the $k$-checkerboard ensemble over $D$.
\end{defi}

When not stated or otherwise clear from context, we assume that $D=\R$ when talking about $k$-checkerboard matrices. We use $w=1$ throughout for simplicity, since only slight alterations are needed to make the results hold for any $w \neq 0$.

For example, a $(2,w)$-checkerboard matrix $A$ would be of the form
\begin{equation}
A\ =\ \begin{bmatrix}
w & a_{0\,1} & w & a_{0\,3} & w & \cdots & a_{0\,N-1}\\
a_{0\,1} & w & a_{1\,2} & w & a_{1\,4} & \cdots & w \\
w & a_{1\,2} & w & a_{2\,3} & w & \cdots & a_{2\,N-1}\\
\vdots & \vdots & \vdots & \vdots & \vdots & \ddots &\vdots \\
a_{0\,N - 1} & w & a_{2\,N - 1} & w & a_{4\,N - 1} & \cdots & w
\end{bmatrix}. \end{equation}


%
%


\subsection{Results}

Let $\nu_{A,N}$ be the empirical spectral measure of a $N\times N$ matrix $A$, where we have normalized the eigenvalues by dividing by $\sqrt{N}$:
\begin{equation}\label{eqn Wigner Spectral Distribution Measure}
\nu_{A, N}\ =\ \frac{1}{N}\sum_{i = 1}^{N} \delta\left(x - \frac{\lambda_{i}}{\sqrt{N}}\right),
\end{equation}
where the $\{\lambda_{i}\}_{i = 1}^{N}$ are the eigenvalues of $A$. Here, we use $A$ and $N$ in the subscript to highlight both the matrix and its size. Wigner's semicircle law states that for many random matrix ensembles, for almost all sequences $\{A_N\}_{N \in \N}$ of $N \times N$ matrices $A_N$,  we have weak convergence of empirical spectral measures $\nu_{A_N,N}$ as $N\to\infty$ to the semicircle measure of radius $R$, $\sigma_R$, which has density
\begin{equation}\label{eqn semi-circle law}
\begin{cases}
\frac{2}{\pi R^2}\sqrt{R^2 - x^{2}} & \text{{\rm if} } |x| \leq R\\
0 & \text{{\rm if} } |x| > R.
\end{cases}
\end{equation}

Note that for $R \neq 1$ the `semicircle' distribution is actually a semi-ellipse with horizontal axis $R$. While one can renormalize the eigenvalues by a constant independent of $N$ to rescale to a semicircle, we will see below that in our setting that constant would depend on $k$. We prefer not to introduce a renormalization dependent on $k$, as it makes no material difference.

For the ensembles mentioned in \S\ref{sec:background} one is able to determine the limiting spectral measure through the method of moments. The situation is more subtle here. As we argue later, the $k$-checkerboard matrices have $k$ eigenvalues of size $N/k$. As the variance of these eigenvalues is of order $k$, for fixed $k$ we see these eigenvalues are well-separated from the $N-k$ eigenvalues that are of order $\sqrt{N}$.  In fact, by using a matrix perturbation approach we are able to establish the following result, which we prove in Appendix \ref{app_tworegimes}:

\begin{theorem}
Let $\{A_N\}_{N\in\N}$ be a sequence of $(k,w)$-checkerboard matrixs. Then almost surely as $N\rightarrow \infty$ the eigenvalues of $A_N$ fall into two regimes: $N-k$ of the eigenvalues are $O(N^{1/2+\epsilon})$ and $k$ eigenvalues are of magnitude $Nw/k+O(N^{1/2+\epsilon})$.
\end{theorem}

We refer to the $N-k$ eigenvalues that are on the order of $\sqrt{N}$ as the eigenvalues in the \textbf{bulk}, while the $k$ eigenvalues near $N/k$ are called the eigenvalues in the \textbf{blip}. See \cite{CHS} for some general results about a class of random matrices exhibiting a different kind of split behavior.

While the presence of these $k$ large eigenvalues prevent us from using one of the standard techniques, the method of moments, to determine the limiting density of the eigenvalues in the bulk, numerics (see Figure \ref{fig example of the blip}) suggest that the limit is a semi-ellipse.



\begin{figure}[H]
\begin{center}
\scalebox{0.7}{\includegraphics{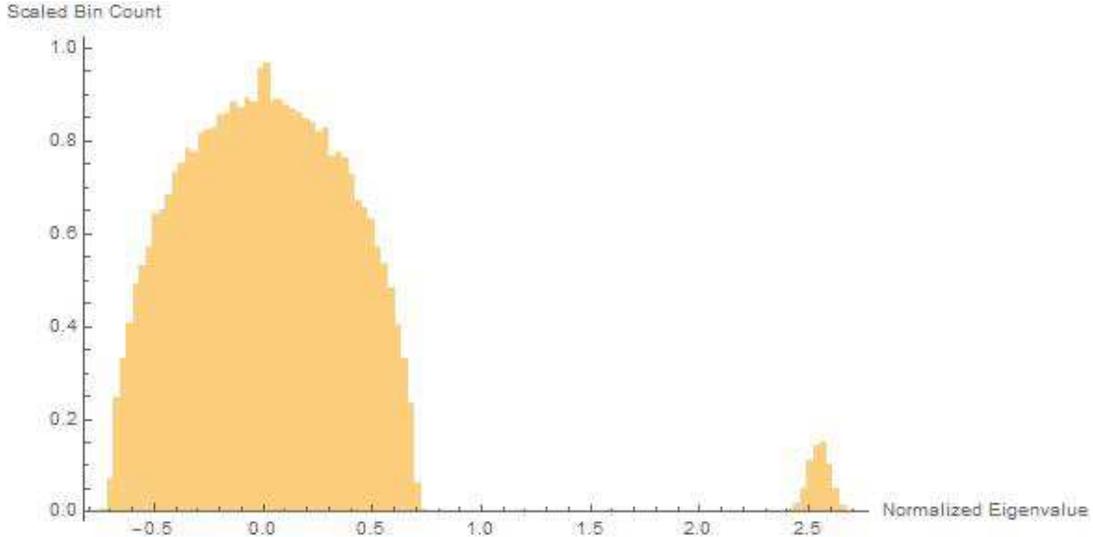}}
\caption{\label{fig example of the blip} A histogram, normalized appropriately to achieve unit mass,
of the scaled eigenvalue distribution for $100 \times 100$ $2$-checkerboard real matrices with $w = 1$ after 500 trials.}
\end{center}
\end{figure}

The following result (see \cite{Tao}) allows us to bypass the complications presented by the small number of large eigenvalues.

\begin{theorem}[\cite{Tao}]\label{stabilityESDrank}
Let $\{\mathcal{A}_N\}_{N \in \N}$ be a sequence of random Hermitian matrix ensembles such that $\{\nu_{\mathcal{A}_N,N}\}_{N \in \N}$ converges weakly almost surely to a limit $\nu$. Let $\{\tilde{\mathcal{A}}_N\}_{N \in \N}$ be another sequence of random matrix ensembles such that $\frac{1}{N}\rank(\tilde{\mathcal{A}}_N)$ converges almost surely to zero. Then $\{\nu_{\mathcal{A}_N+\tilde{\mathcal{A}}_N,N}\}_{N\in\N}$ converges weakly almost surely to $\nu$. 
\end{theorem}

Taking $\tilde{\mathcal{A}}_N$ to be the fixed matrix with entries $m_{ij}=1_{i \equiv j \pmod{k}}$ implies that the limiting spectral distribution of the $k$-checkerboard ensemble as defined previously with $w=1$ is the same as the limiting spectral distribution of the ensemble with $w=0$, which does not have the $k$ large blip eigenvalues (for the remainder of this paper, $A_N$ always refers to an $N \times N$ matrix). This overcomes the issue of diverging moments.

\begin{theorem}\label{bulk_limit}
Let $\{A_N\}_{N \in \N}$ be a sequence of real $N \times N$ $k$-checkerboard matrices. Then the empirical spectral measures $\nu_{A_N,N}$ converges weakly almost surely to the semicircle distribution.
\end{theorem}


The proof is by standard combinatorial arguments. We give the details in Appendix \ref{app_bulk}.

On the other hand, the blip is where the vast majority of interesting behavior and technical challenges are encountered. We begin with some heuristic arguments which give intuition for how the blip arises and behaves.

Firstly, recall that a matrix $A$ for which the sum of all entries in any given row is equal to some fixed $d$ has the trivial eigenvalue $d$ with eigenvector $(1,1,\ldots,1)^T$. For a matrix in the $N \times N$ $k$-checkerboard ensemble, the sum of the $i$\textsuperscript{th} row is equal to $N/k+\sum_{j=1}^N a_{ij}$ where the $a_{ij}$ are i.i.d. with mean $0$ and variance $1$. This is approximately $N/k$, so heuristically there should be an eigenvector very close to $(1,1,\ldots,1)^T$ with eigenvalue roughly $N/k$. Similarly, there are $k-1$ other eigenvalues of size approximately $N/k$ with eigenvectors close to the one described previously with some additional periodic sign changes.

Hence the blip may be thought of as deviations about the trivial eigenvalues. The surprising result of this paper is that these deviations, while seemingly quite different from the eigenvalue distributions of classical random matrix theory, in fact have the same distribution as the eigenvalues of the following $k \times k$ random matrix sub-ensemble of the classical Gaussian Orthogonal Ensemble (GOE).

\begin{defi} \label{def hollow GOE}
The \textbf{hollow Gaussian Orthogonal Ensemble} is given by $A = (a_{ij}) = A^T$ with
\begin{equation}
a_{ij} = \begin{cases}
\mathcal{N}_{\R}(0, 1) & \text{{\rm if} } i \neq j \\
0 & \text{{\rm if} } i = j.
\end{cases}
\end{equation}
\end{defi}

The spectral distribution of the $2 \times 2$ hollow GOE is Gaussian (see Proposition \ref{prop_k2gaussian}), and in the $k \rightarrow \infty$ limit the eigenvalue distribution is a semicircle by standard GOE arguments. For larger finite $k$ we see an interesting sequence of distributions which interpolate between the Gaussian and the semicircle, similarly to the results in \cite{KKMSX} for block circulant matrices. The first few are shown in Figures \ref{figure:hist2hist3} and \ref{figure:hist4hist16}.

\begin{figure}[H]
\centering\includegraphics[scale=.7]{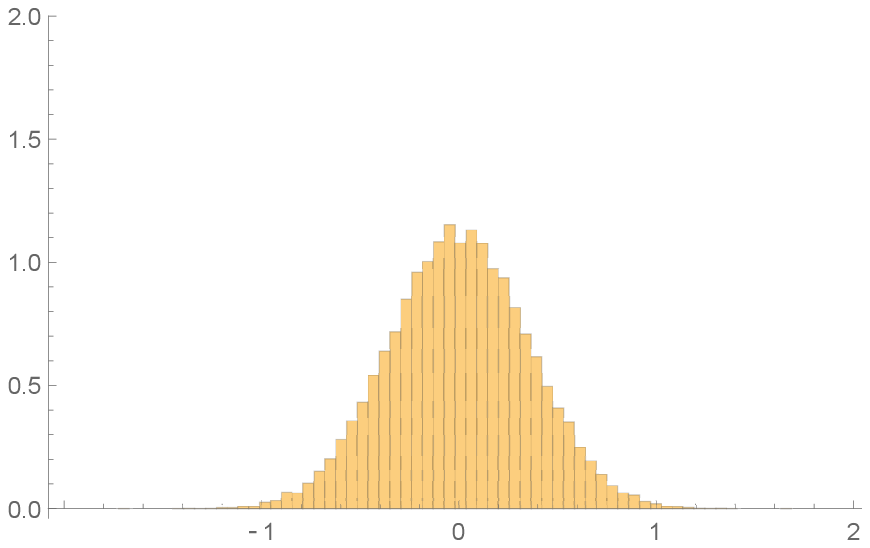}\ \centering\includegraphics[scale=.7]{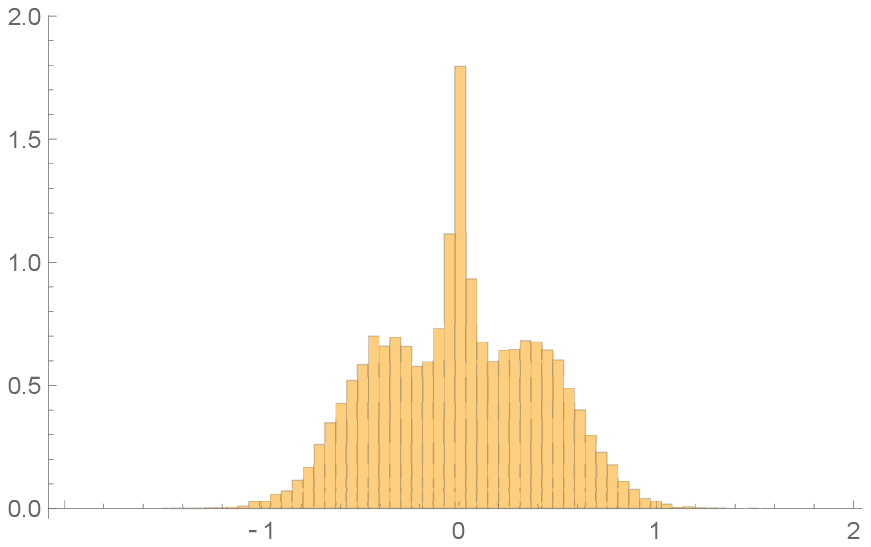}
\caption{(Left) Histogram of eigenvalues of 32000 $2\times2$ hollow GOE matrices. (Right) Histogram of eigenvalues of 32000 $3\times3$ hollow GOE matrices.}\label{figure:hist2hist3}
\centering\includegraphics[scale=.7]{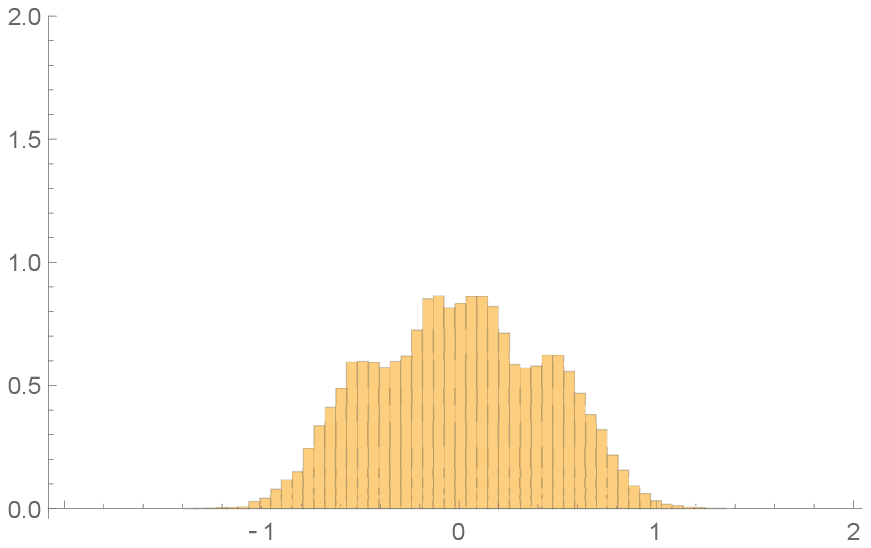}\ \centering\includegraphics[scale=.7]{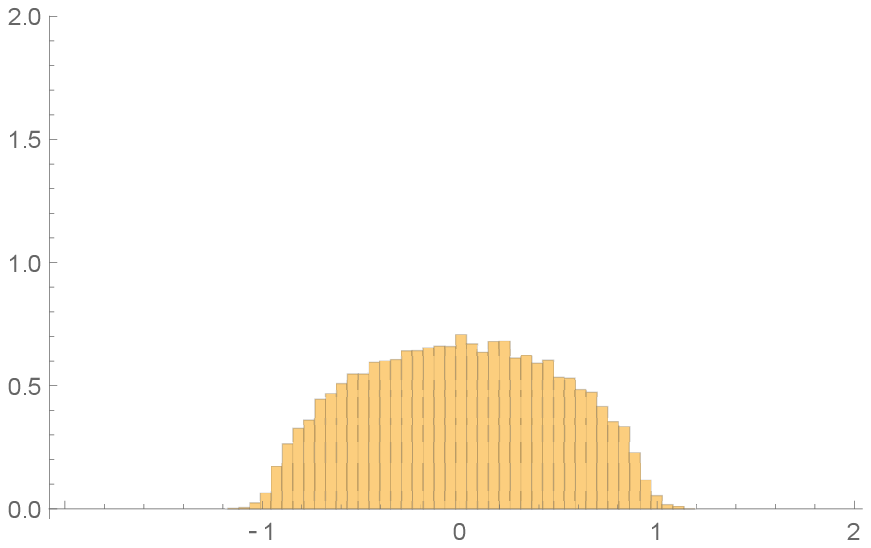}
\caption{(Left) Histogram of eigenvalues of 32000 $4\times4$ hollow GOE matrices. (Right) Histogram of eigenvalues of 32000 $16\times16$ hollow GOE matrices.}\label{figure:hist4hist16}
\end{figure}

Computing this distribution poses substantial challenges. Ideally, we would like to define a weighted blip spectral measure which takes into account only the eigenvalues of the blip and not the bulk. Naively, one could multiply the empirical spectral measure by some smooth cutoff function of the form $\mathbf{1}_{[N/k-\delta(N),N/k+\delta(N)]}$ for $\delta(N)$ growing appropriately to capture all of the blip and neglect the bulk in the limit. However, with such a weighting function we cannot use the eigenvalue-trace formula to reduce the problem to combinatorics on products of matrix entries in the standard way. The next reasonable possibility is to try Taylor expanding a nice cutoff function, for then each expected moment is of the form
\begin{equation}
\E\left[\sum_{i=0}^\infty c_i p_i(\l_1,\ldots,\l_N)\right],
\end{equation}
where $p_i$ is the power sum symmetric polynomial of degree $i$ and $\l_j$'s are the eigenvalues. Unfortunately, Taylor series convergence and limit-switching issues make this approach untenable.

Hence, we are led to use a polynomial weighting function. No polynomial of fixed degree is a sufficiently good approximation to a smooth cutoff function, so we use a sequence of polynomials of degree increasing with the matrix size $N$ so that in the limit we mimic a smooth cutoff function. Specifically, let
\begin{equation}
f_{n}(x) := x^{2n}(x-2)^{2n}
\end{equation}
Thus we alter the standard empirical spectral measure in the following way to capture the blip.

\begin{defi}\label{def_empirical_blip_measure}
The \textbf{empirical blip spectral measure} associated to an $N \times N$ $k$-checkerboard matrix $A$ is
\begin{equation}\label{eqn modified spectral measure for the blip}
\mu_{A, N}\ :=\ \frac{1}{k}\sum_{\lambda \text{ eigenvalue of }A} f_{n(N)}\left(\frac{k \lambda}{N}\right) \delta\left(x - \left(\lambda - \frac{N}{k} \right)  \right)
\end{equation}
where $n(N)$ is a function for which there exists some $\eps$ so that $N^\eps \ll n(N) \ll N^{1-\eps}$.
\end{defi}

At a blip eigenvalue $\l \approx N/k$, we have $f_n\left(\frac{\l}{N/k}\right) \approx 1$; because the standard deviation of the bulk eigenvalues $\l'$ is on the order of $\sqrt{N}$, $f_n\left(\frac{\l'}{N/k}\right)\approx 0$ for any bulk eigenvalue $\l'$. Because $f_n(1)=1$, $f(0)=0$, and $f_n'(1)=0=f_n'(0)=\cdots= f_n^{(2n-1)}(0)$, the bulk eigenvalues are given weight roughly $0$ and the blip eigenvalues are all given weight roughly $1$, and small deviations about these weights disappear in the limit.

\begin{rek}
The authors experimented with several other sequence of polynomials and all give the same end results under some suitable conditions, but this one simplifies computations. Furthermore, it is nonnegative, ensuring that the empirical blip spectral measure is actually a measure. It is almost a probability measure, i.e. for a typical matrix $\mu_{A,N}(\R)$ is close to $1$. To make $\mu_{A,N}$ a probability measure we would need to divide by the sum of the weights associated to the eigenvalues, but the expected value of this quotient is intractable, so we instead divide by $k$.
\end{rek}

Definition \ref{def_empirical_blip_measure} finally allows reduction to tractable combinatorics. Interestingly, this combinatorics reduces back to random matrix theory, yielding convergence in expectation of the moments of the weighted blip spectral measure of the $k$-checkerboard matrix ensemble to those of the $k \times k$ hollow GOE. However, we cannot show almost-sure weak convergence of measures by standard arguments because (a) due to the weighting function, the blip empirical spectral measure is no longer a probability measure, and (b) the number of eigenvalues in the blip is fixed so there are not enough to average over. We modify the moment convergence theorem to overcome the first difficulty, and average over the eigenvalues of multiple independent matrices to overcome the second.

We now state this result formally.

\begin{defi}\label{def_ave_blip_fin}
Fix a function $g: \N \rightarrow \N$. The \textbf{averaged empirical blip spectral measure} associated to a $g(N)$-tuple of $N \times N$ $k$-checkerboard matrices $(A_N^{(1)},A_N^{(2)},\dots,A_N^{(g(N))})$ is
\begin{equation}
\mu_{N,g,A_N^{(1)},A_N^{(2)},\dots,A_N^{(g(N))}}\ :=\ \frac{1}{g(N)}\sum_{i=1}^{g(N)} \mu_{A_N^{(i)},N}.
\end{equation}
\end{defi}

\begin{thm}\label{thm_main_goe}
Let $g: \N \rightarrow \N$ be such that there exists a $\delta>0$ for which $g(N) \gg N^\delta$. Let $A^{(i)}=\{A^{(i)}_N\}_{N \in \N}$ be sequences of fixed $N \times N$ matrices, and let $\overline{A}=\{A^{(i)}\}_{i \in \N}$ be a sequence of such sequences.
Then, as $N\to\infty$, the averaged empirical blip spectral measures $\mu_{N,g,A_N^{(1)},A_N^{(2)},\dots,A_N^{(g(N))}}$ of the $k$-checkerboard ensemble over $\R$ converge weakly almost-surely to the measure with moments equal to the expected moments of the standard empirical spectral measure of the $k\times k$ hollow Gaussian Orthogonal Ensemble.
\end{thm}

One can also naturally define the \emph{hollow Gaussian Unitary Ensemble} and the \emph{hollow Gaussian Symplectic Ensemble} by extending Definition \ref{def hollow GOE} to complex valued matrices comprised of complex Gaussians and quaternion valued matrices comprised of quaternion Gaussians, respectively. In Theorem \ref{thm quaternion GSE case} we obtain analogous results to Theorem \ref{thm_main_goe}, connecting the limiting blip spectral measure of the $k$-checkerboard ensembles over $\C$ and $\H$ to the empirical spectral measures of the hollow GUE and hollow GSE, respectively.

In \S\ref{sec:bulk} we prove our claims concerning the eigenvalues in the bulk, then turn to the blip spectral measure in  \S\ref{sec The Moments of the Blip Spectral Measure} (and the mentioned generalizations in \S\ref{sec:genCandH}). We then prove results on the convergence to  the limiting spectral measure in \S\ref{sec:convergence}.


\section{The Bulk Spectral Measure}\label{sec:bulk}

In this section we establish that the limiting bulk measure for $k$-checkerboard matrices follows a semi-circle law. We denote by $\mu^{(m)}$ the $m$\textsuperscript{th} moment of the measure $\mu$.

\begin{theorem}\label{thm_main_bulk}
Let $\{A_N\}_{N \in \N}$ be a sequence of $N \times N$ $(k,1)$-checkerboard matrices, and let $\nu_{A_N}$ denote the empirical spectral measure, then $\nu_{A_N}$ converges weakly almost surely to the Wigner semicircle measure $\sigma_R$ with radius
\begin{equation}
R\ =\ 2\sqrt{1-1/k}.
\end{equation}
\end{theorem}

One common tool used to study the limiting spectral density of a matrix ensemble is the method of moments. However, this method cannot be applied directly to the study of checkerboard matrices when studying the bulk regime because the limiting expected moments of the empirical spectral measure do not exist. For a proof of their divergence, see Proposition \ref{prop:bulk_divergence} in the appendix. 
The following result overcomes this difficulty by allowing us to treat the $w$ entries as $0$.

\begin{theorem}\label{stabilityESDrankINPAPER}
\cite{Tao}
Let $\{\mathcal{A}_N\}_{N \in \N}$ be a sequence of random Hermitian matrix ensembles such that $\{\nu_{\mathcal{A}_N,N}\}_{N \in \N}$ converges weakly almost surely to a limit $\nu$. Let $\{\tilde{\mathcal{A}}_N\}_{N \in \N}$ be another sequence of random matrix ensembles such that $\frac{1}{N}\rank(\tilde{\mathcal{A}}_N)$ converges almost surely to zero. Then $\{\nu_{\mathcal{A}_N+\tilde{\mathcal{A}}_N,N}\}_{N\in\N}$ converges weakly almost surely to $\nu$.
\end{theorem}

We now use the method of moments to establish the result for $(k,0)$-checkerboard matrices. The main work is using combinatorics to establish convergence of the expected moments. The remaining arguments  establishing almost sure weak convergence are standard and may be found in Appendix \ref{app_bulk}.

\begin{lemma}\label{lem:avgmomentsbulk}
The expected moments of the bulk empirical spectral measure taken over $A_N$ in the $N \times N$ $(k,0)$-checkerboard ensemble converge to the moments of the Wigner semicircle distribution $\sigma_R$ with radius $R=2\sqrt{1-1/k}$
\begin{equation}
\E\left[\nu^{(\ell)}_{A_N}\right] \rightarrow \sigma_R^{(\ell)}
\end{equation}
as $N \rightarrow \infty$.
\end{lemma}


\begin{proof}
We have immediately from the eigenvalue-trace lemma and linearity of expectation that
\begin{equation}
\E\left[\nu^{(\ell)}_{A_N}\right]\ = \ \frac{1}{N^{\ell/2+1}}\sum_{1\le i_1,\ldots,i_\ell\le N} \E \left[a_{i_1i_2}\cdots a_{i_{\ell-1}i_{\ell}}a_{i_\ell i_1}\right].
\end{equation}

Each term in the sum is associated to a sequence $\mathbf{I}=i_1i_2\ldots i_\ell i_1$. Each sequence corresponds to a closed walk on the complete graph with vertices labeled by the elements of the set $\{i_1,...,i_\ell\}$ by giving the order in which the vertices are visited. Define the \textit{weight} of such a sequence $\mathbf{I}$ to be the number of distinct entries of $\mathbf{I}$. If the weight of a walk is greater than $\ell/2+1$, the walk contributes nothing to the sum because the expectation of some entry is independent of the rest and its expectation is $0$.

The sequences of weight less than $\ell/2+1$ contribute $o(N^{\ell/2+1})$ to the sum. This is because the sequences may be partitioned into a finite number of equivalence classes by the isomorphism class of the corresponding walk. An isomorphism class of weight $t$ then gives rise to $O(N^t)$ walks of weight $t$ by choosing labels for the distinct nodes in any such walk. 

The sequences of weight $\ell/2+1$ require a finer analysis. When $\ell$ is odd,
 the expectation associated to each such sequence is $0$. When $\ell$ is even,
 the walk corresponding to such a sequence visits $\ell/2+1$ nodes and traverses $\ell/2$ distinct edges. Hence as the walk is connected, it is a tree. Moreover, each walk may be rooted by associating the initial node of the walk to the root. As is well known, there are $C_{\ell/2}$ rooted trees on $\ell/2+1$ nodes, where $C_\ell$ is the $\ell$\textsuperscript{th} Catalan number. We may then label the nodes in the tree in such a way that no two adjacent nodes have the same congruence class in $N^{\ell / 2 + 1} \pfrac{k - 1}{k}^{\ell / 2} + o\pez{N^{\ell / 2 + 1}}$ ways. Writing $\zeta_{\mathbf{I}}$ for $a_{i_1i_2}\cdots a_{i_{\ell-1}i_\ell}a_{i_\ell i_1}$, we have

\begin{align}
\E\left[\nu^{(\ell)}_{A_N}\right]\ &\ =\  \ \frac{1}{N^{\ell/2+1}}\left(\sum_{\text{weight } \mathbf{I} <\ell/2+1}\E\left[\zeta_I\right] +\sum_{\text{weight } \mathbf{I} =\ell/2+1}\E\left[\zeta_I\right] +\sum_{\text{weight } \mathbf{I} >\ell/2+1}\E\left[\zeta_I\right]  \right)\nonumber\\
&\ =\ \frac{1}{N^{\ell/2+1}}\left(o(N^{\ell/2+1})+ C_{\ell/2}\left( N^{\ell / 2 + 1} \pfrac{k - 1}{k}^{\ell / 2} + o\pez{N^{\ell / 2 + 1}}\right)+0 \right)\nonumber\\
&\ =\ C_{\ell/2}\pfrac{k - 1}{k}^{\ell / 2}+o(1) \label{eqn_bulk_exp}
\end{align}

Hence
\begin{equation}
\lim_{N\rightarrow \infty}\E\left[\nu^{(\ell)}_{A_N}\right] \ = \ \begin{cases}
\pez{\frac{R}{2}}^{\ell}C_{\ell/2}& \text{ if } \ell \text{ is even}\\
0&\text{ otherwise,}
\end{cases}
\end{equation}
which are the moments of the semicircle distribution of radius $R$.
\end{proof}




\section{The Blip Spectral Measure}\label{sec The Moments of the Blip Spectral Measure}

To appropriately modify the measure in \eqref{eqn Wigner Spectral Distribution Measure}, we weight by the polynomial
\begin{equation}\label{eqn weighting function}
f_{n}(x)\ =\ x^{2n}(x-2)^{2n}
\end{equation}
and study the following spectral measure.

\begin{defi}
The \textbf{empirical blip spectral measure} associated to an $N \times N$ $k$-checkerboard matrix $A$ is
\begin{equation}
\mu_{A, N}\ :=\ \frac{1}{k}\ \sum_{\lambda\ \text{{\rm an\ eigenvalue\ of\ }}A}\ f_{n(N)}\left(\frac{k \lambda}{N}\right) \delta\left(x - \left(\lambda - \frac{N}{k} \right)  \right),
\end{equation}
where $n(N)$ is a function for which there exists some $\eps$ so that $N^\eps \ll n(N) \ll N^{1-\eps}$; the particular choice is not important as long as these conditions are satisfied.
\end{defi}


The modified spectral measure of Definition \ref{def_empirical_blip_measure} weights eigenvalues within the blip by almost exactly 1, due to the scaling, and those in the bulk are weighted by almost exactly zero.  We shift the eigenvalues by subtracting roughly mean of the blip in order to center the blip rather than the bulk. This does not truly center the blip, but causes the center to remain fixed as $N \rightarrow \infty$; we compute the limiting moments of this measure and center later. 

First, we explicitly derive a formula for the expected $m$\textsuperscript{th} moment of the blip spectral measure given in \eqref{eqn modified spectral measure for the blip}, where the expectation is taken over the $N \times N$ $k$-checkerboard ensemble.

\begin{lem}\label{lem Moments of blip spectral measure}
The expected $m$\textsuperscript{th} moment of the blip empirical spectral measure, $\mu_{A,N}$, is
\begin{equation}\label{eqn moments of blip spectral measure}
\E[\mu_{A,N}^{(m)}]\ =\ \frac{1}{k}\left(\frac{k}{N}\right)^{2n} \sum_{j=0}^{2n} \binom{2n}{j} \sum_{i=0}^{m+j} \binom{m+j}{i} \left(-\frac{N}{k}\right)^{m-i} \etr A^{2n+i}
\end{equation}
where both expectations are taken over the $N \times N$ $k$-checkerboard ensemble.
\end{lem}

\begin{proof}
We have
\begin{align}
\E[\mu_{A,N}^{(m)}] &\ =\ \frac{1}{k}\E\left[\sum_{\gl}f\left(\frac{\gl}{N/k}\right) \left(\gl-\frac{N}{k}\right)^m\right] \nonumber\\
&\ =\ \frac{1}{k}\E\left[\sum_{\gl}\left(\frac{k\gl}{N}\right)^{2n}\sum_{j=0}^{2n} \binom{2n}{j} (-1)^j \left(\frac{k}{N}\right)^j \left(-\frac{N}{k}\right)^j\sum_{i=0}^{m+j} \binom{m+j}{i}\left(-\frac{N}{k}\right)^{m-i}\gl^i \right]\nonumber\\
&\ =\ \frac{1}{k}\left(\frac{k}{N}\right)^{2n} \sum_{j=0}^{2n} \binom{2n}{j}\sum_{i=0}^{m+j} \binom{m+j}{i}\left(-\frac{N}{k}\right)^{m-i}\E\left[\sum_{\gl} \gl^{2n+i}\right] \nonumber\\
&\ =\ \frac{1}{k}\left(\frac{k}{N}\right)^{2n} \sum_{j=0}^{2n} \binom{2n}{j}\sum_{i=0}^{m+j} \binom{m+j}{i}\left(-\frac{N}{k}\right)^{m-i}\etr A^{2n+i},
\end{align}
where the first equality comes from straightforward algebra using binomial expansion and the last equality comes from the Eigenvalue-Trace Lemma.
\end{proof}

Now, recall that
\begin{equation}
\etr M^n \ =\ \sum_{1 \leq i_1, \ldots, i_n \leq N} \E[m_{i_1 i_2}m_{i_2 i_3} \cdots m_{i_n i_1}].
\end{equation}

We refer to terms $\E[m_{i_1 i_2}m_{i_2 i_3} \cdots m_{i_n i_1}]$ as \textbf{cyclic products} and $m$'s as entries of cyclic products. By Lemma \ref{lem Moments of blip spectral measure}, it suffices to understand the cyclic products making up $\etr A^{2n+i}$, which reduces to a combinatorics problem of understanding the contributions of different cyclic products. We develop the following vocabulary to classify types of cyclic products according to the aspects of their structure that determine overall contributions. 

\begin{defi}\label{def of a Block}
A \textbf{block} is a set of adjacent $a$'s surrounded by $w$'s in a cyclic product, where the last entry of a cyclic product is considered to be adjacent to the first. We refer to a block of length $\ell$ as an $\ell$-block or sometimes a block of size $\ell$.
\end{defi}

\begin{defi}\label{def of a Configuration}
A \textbf{configuration} is the set of all cyclic products for which it is specified (a) how many blocks there are, and of what lengths, and (b) in what order these blocks appear. However, it is not specified how many $w$'s there are between each block.
\end{defi}

\begin{exa}
The set of all cyclic products of the form $w\cdots waw \cdots waaw \cdots waw \cdots w$, where each $\cdots$ represents a string of $w$'s and the indices are not yet specified, is a configuration.
\end{exa}

\begin{defi}\label{def of a Class}
Let $S$ be a multiset of natural numbers. An $S$\textbf{-class}, or class when $S$ is clear from context, is the set of all configurations for which there exists a unique $s$-block for every $s \in S$ counting multiplicity.  In other words, two configurations in the same class must have the same blocks but they may be ordered differently and have different numbers of $w$'s between them.
\end{defi}

When we speak of the \emph{contribution} of a configuration or class to $\etr A^{2n+i}$, we assume that the length of the cyclic product is fixed at $2n+i$. The reason that the length of the cyclic product is suppressed in our notation is because $n(N)$ varies with $N$ and we wish to consider the contribution of a configuration or class as $N \rightarrow \infty$. 

\begin{defi}\label{def of a Matching}
Given a configuration, a \textbf{matching} is an equivalence relation $\sim$ on the $a$'s in the cyclic product which constrains the ways of indexing (see Definition \ref{def of an Indexing}) the $a$'s as follows: an indexing of $a$'s conforms to a matching $\sim$ if, for any two $a$'s $a_{i_{\ell},i_{\ell+1}}$ and $a_{i_{t},i_{t+1}}$, we have $\{i_\ell,i_{\ell+1}\}=\{i_t,i_{t+1}\}$ if and only if $a_{i_{\ell}i_{\ell+1}} \sim a_{i_{t},i_{t+1}}$. We further constrain that each $a$ is matched with at least one other by any matching $\sim$.
\end{defi}

\begin{rem}
Noting that the $a_{ij}$ are drawn from a mean-$0$ distribution, any matching with an unmatched $a$ would not contribute in expectation, hence it suffices to only consider those with the $a$'s matched at least in pairs.
\end{rem}

\begin{exa}
Given a configuration $a_{i_1i_2}w_{i_2i_3}a_{i_3i_4}w_{i_4i_5}a_{i_5i_6}w_{i_6i_7}a_{i_7i_8}w_{i_8i_1}$ (the indices are not yet specified because this is a configuration), if $a_{i_1i_2} \sim a_{i_5i_6}$ we must have either $i_1=i_5$ and $i_2=i_6$ or $i_1=i_6$ and $i_2=i_5$.
\end{exa}

\begin{defi}\label{def of an Indexing}
Given a configuration, matching, and length of the cyclic product, then an \textbf{indexing} is a choice of
\begin{enumerate}
\item the (positive) number of $w$'s between each pair of adjacent blocks (in the cyclic sense), and
\item the integer indices of each $a$ and $w$ in the cyclic product.
\end{enumerate}
\end{defi}

Two comments on these definitions are in order.

\begin{rem}
 It is very important to note that the definitions of class, configuration, and matching do \emph{not} fix the length of the cyclic product and hence we may consider their contribution as $n(N)$ grows; however, because the length of the product directly affects the number of indexings, we must take it into account when summing over them.
\end{rem}

\begin{rem}
Note that the choice of indexings is constrained by the configuration as well as the matching, because entries $a_{ij}$ have $i \not \equiv j \pmod{k}$ and $w_{ij}$ have $i \equiv j \pmod{k}$. This is important later.
\end{rem}

With the above vocabulary, we have
\begin{equation}
\etr A^{\eta}\ =\ \sum_{\substack{S\text{-classes} \\ C} }\; \sum_{\substack{\text{configurations} \\ \mathscr{C} \in C} }\; \sum_{\substack{\text{matchings} \; M }}\; \sum_{\substack{\text{indexings } \\ I \text{ given } M,\mathscr{C},\eta} } \E[\Pi]
\end{equation}
where $\Pi$ is the cyclic product given by the choice of indexing.

The following lemma allows us to determine which $S$-classes contribute to the trace terms in \eqref{eqn moments of blip spectral measure} and which contributions become insignificant in the limit as $N \rightarrow \infty$.

\begin{lem}\label{lem_contributions}
In the limit as $N \rightarrow \infty$, the only classes which contribute are those with only $1$- or $2$-blocks, $1$-blocks are matched with exactly one other $1$-block, and both $a$'s in any $2$-block are matched with their adjacent entry and no others.
\end{lem}

\begin{proof}
Fix the number of blocks $\beta$ in the classes we consider. We refer to the power of $N$ in the contribution of our class as its number of degrees of freedom; each degree of freedom corresponds to the choice of an index in the cyclic product. For a fixed configuration, we consider the total number of degrees of freedom lost by the constraints placed by a matching. Because we have fixed the number of blocks, we may then talk about the average number of degrees of freedom lost per block. Given a $2$-block $a_{ij}a_{j\ell}$, matching the two $a$'s constrains $\ell=i$ and hence loses one degree of freedom; if two singletons $a_{ij}$ and $a_{t\ell}$ are matched then $\{i,j\}=\{t,\ell\}$ so two degrees of freedom are lost. Therefore if all blocks have size $1$ or $2$ and the matchings are as in the hypotheses, then one degree of freedom per block is lost when averaged over all blocks, no matter the configuration or length of the cyclic product. Thus classes in which more than one degree of freedom per block is lost do not contribute in the $N \rightarrow \infty$ limit, so it suffices to show any classes and matchings which are not as specified in the lemma statement lose more than one degree of freedom per block.

Fix a configuration $\mathcal{C}$ with $\alpha$ $a$'s and a matching $\sim$. Then $\sim$ partitions the $a$'s in $\mathcal{C}$ into equivalence classes $T_1,\ldots,T_s$. If there were no matching restrictions, only the restriction that the first index of an $a$ matches the second index of the last one, then the number of degrees of freedom from choosing the indices of the $a$'s would be
\begin{equation}
M\ =\ \sum_{\text{blocks }b}(\text{len}(b)+1)\ =\ \beta + \alpha.
\end{equation}
Let $F$ be the actual number of degrees of freedom from choosing the indices of the $a$'s, given our configuration and matching. Naively, we may choose two indices for each matching class $T_1,\ldots,T_s$, but then there may be restrictions from $a$'s from different matching classes being adjacent that cause a loss of degrees of freedom. Letting $c$ be the number of degrees of freedom lost to such crossovers, the number of degrees of freedom we have is $2s-c$. Then the number of degrees of freedom lost per block is
\begin{equation}
\frac{M-(2s-c)}{\beta}\ =\ 1 + \frac{\alpha + c - 2s}{\beta}.
\end{equation}
It thus suffices to show that our configuration and matching are of the form specified in the lemma statement if and only if $\alpha + c - 2s = 0$, or equivalently $\frac{\alpha + c}{s} = 2$, and $\frac{\alpha + c}{s} >2$ for any other configuration and matching. The forward direction was proven in the beginning of this proof.

For the backward direction, because $|T_i| \geq 2$ for all $i$ by the definition of matching, we immediately have $\frac{\alpha}{s} \geq 2$. If there is some $T_i$ with $|T_i|>2$ then we have $\frac{\alpha}{s}>2$, and if there exist $i,j$ such that an $a$ from $T_i$ is adjacent to an $a$ from $T_j$ then we have $\frac{c}{s}>0$. Therefore if $\frac{\alpha+c}{s}=2$ then there is no $T_i$ with $|T_i|>2$ or $i,j$ such that an $a$ from $T_i$ is adjacent to an $a$ from $T_j$, i.e., the $a$'s are matched in pairs and no unmatched $a$'s are adjacent. This proves the lemma.
\end{proof}

We now explicitly compute the contributions of each of these classes.

\begin{prop}\label{prop_polygon}
The total contribution to $\etr A^\eta$ of an $S$-class $C$ with $m_1$ $1$-blocks and $(|S|-m_1)$ $2$-blocks
\begin{equation}\label{eq_12class_contribution}
p(\eta)\binom{|S|}{m_1}(k-1)^{|S|-m_1}\ektr B^{m_1}  \left(\pfrac{N}{k}^{\eta-|S|}+O\left(\pfrac{N}{k}^{\eta-|S| - 1} \right)\right)
\end{equation}
where
\begin{equation}
p(\eta)\ =\ \frac{\eta^{|S|}}{|S|!} + O(\eta^{|S| - 1})
\end{equation}
and the expectation $\ektr B^{m_1}$ is taken over the $k \times k$ hollow GOE as defined in Definition \ref{def hollow GOE}.
\end{prop}

\begin{proof}
Let $\mathcal{A}=m_1+2(m-m_2)$ denote the number of $a$'s in $C$. Let $p(\eta)$ be the number of ways to arrange $|S|$ blocks and $(\eta-\mathcal{A})$ $w$'s into a cyclic product of length $\eta$, where the blocks are taken to be indistinguishable. We first compute $p(\eta)$.   We may think of the configurations in $C$ by bijectively identifying them with the set of $(\eta - (\mathcal{A} - |S|))$-gons with $|S|$ non-adjacent vertices labeled by $a$ (these correspond to blocks of any size, not just $1$-blocks) and the rest labeled by $w$.  Each $a$ vertex corresponds to a particular block and each $w$ vertex corresponds to a $w$ in the configuration; these are on a polygon rather than a straight line because the first and last entry of a cyclic product are considered adjacent (see Definition \ref{def of a Block}).

We may calculate $p(\eta)$ by first examining all possible choices of $\binom{\eta - (\mathcal{A} - |S|)}{|S|}$ distinct vertices, then subtracting off all the cases for which at least one pair of the vertices selected are adjacent.  If one pair is adjacent, then--with no other restrictions placed upon the other vertices--there are $\eta - (\mathcal{A} - |S|)$ possible locations for the $2$-block to be placed.  This leaves $\binom{\eta - \mathcal{A} + |S| - 2}{|S| - 2}$ possible locations for the remaining labels.  As such, the term that must be subtracted off has degree in $\eta$ strictly less than $|S|$.  Hence
\begin{equation}
p(\eta)\ =\ \binom{\eta - (\mathcal{A} - |S|)}{|S|} +  O(\eta^{|S| - 1})\ =\ \frac{\eta^{|S|}}{|S|!} + O\left(\eta^{|S| - 1}\right).
\end{equation}


Having specified the locations of the blocks, there are $\binom{|S|}{m_1}$ ways to choose which locations have a $1$-block and which have a $2$-block.

By Definition \ref{defn:checkerboard}, we have that for any entry $a_{ij}$, $i \not \equiv j \mod{k}$, and for any entry $w_{ij}$, $i \equiv j \mod{k}$. We consider what conditions this places on the indices in a given configuration.

\begin{exa}\label{ex congruence}
Consider the configuration
\begin{equation}
\cdots a_{i_1i_2}w_{i_2i_3}w_{i_3i_4}a_{i_4i_5}a_{i_5i_4}w_{i_4i_6}a_{i_6i_7} \cdots.
\end{equation}
Then we have
\begin{equation}
i_2 \equiv i_3 \equiv i_4 \equiv i_6 \pmod{k}
\end{equation}
and these are not congruent to $i_1,i_5$ or $i_7$ mod $k$.
\end{exa}

We thus see that the congruence class of the second index of a $1$-block determines the congruence classes of the indices of the string of $w$'s to its right. Similarly, the leftmost index $i$ of a matched $2$-block $a_{ij}a_{ji}$ determines the rightmost index.
Thus the congruence class modulo $k$ of the second index of a $1$-block propagates through $w$'s and $2$-blocks, and hence determines the congruence class modulo $k$ of the first index of the next $1$-block, where `next' is taken in the cyclic sense for the last $1$-block in the cyclic product.



We now claim that the number of ways to choose congruence classes of the indices of the $1$-blocks, such that there exists a consistent choice of indices for the other entries given the constraints discussed above, is $\ektr B^{m_1}$. First, note that by the above considerations, the number of ways to choose congruence classes mod $k$ of the indices of the $1$-blocks is equal to the number of ways to choose indices of the cyclic product $\kmat_{i_1i_2}\kmat_{i_2i_3}\cdots \kmat_{i_{m_1}i_1}$ with $i_1,\ldots,i_{m_1} \in \{1,\ldots,k\}$ under the restriction $i_j \neq i_{j+1}$ for all $j$.

However, there are two restrictions on our choices of indices. Firstly given any pair of congruence classes mod $k$, any contributing cyclic product must have an even number of $a$'s with both indices coming from that pair of congruence classes, because the $a$'s must be matched in pairs by Lemma \ref{lem_contributions}. Secondly, if there are more than two $1$-block $a$'s with indices from the same pair of congruence classes, then there is a choice as to how to match them\footnote{Recall that we have specified that the indices come from the same congruence class, but we must still specify pairs of $a$'s with indices \emph{actually equal}.}.  Specifically, if there are $2q$ $a$'s with indices from the same pair of congruence classes, then there are $(2q - 1)!!$ ways to match them into pairs.

This means that if we have $q$ $a$'s with indices from the same pair of congruence classes, then there are $0$ ways to get a contributing matching if $q$ is odd and $(q-1)!!$ ways if $q$ is even. But these are exactly the moments of a Gaussian, so given a configuration, the number of ways to specify the congruence classes of the $1$-blocks and specify a matching which will contribute in the limit is
\begin{equation}
\sum_{1 \leq i_1,\ldots,i_r \leq k \text{ distinct}} \E[\kmat_{i_1i_2}\kmat_{i_2i_3}\cdots \kmat_{i_ki_1}]
\end{equation}
with each $\kmat_{ij}\sim \mathcal{N}(0,1)$ i.i.d. under the restriction that $\kmat_{ij} = \kmat_{ji}$ and $\kmat_{ii} = 0$ for all $i$. This is the $k \times k$ hollow GOE as defined in Definition \ref{def hollow GOE}

Finally, after specifying these congruence classes, the congruence classes of the indices of the $w$'s and the outer indices $i$ of pairs $a_{ij}a_{ji}$ are determined as argued previously. However, there are still $k-1$ possible choices of congruence class for each inner index $j$ in each $2$-block, because the congruence class of the inner index $j$ must be different from that of the outer one, which is already determined. Therefore there are $(k-1)^{|S|-m_1}$ ways to choose these congruence classes. After all congruence classes are determined, there are $N/k$ choices for each index. However, because there are $|S|$ blocks, by the proof of Lemma \ref{lem_contributions} there are $|S|$ indices which are determined by another. Therefore the contribution from actually specifying the indices is $\pfrac{N}{k}^{\eta - |S|}$.
Therefore, the contribution from choosing the locations of blocks, then locations of $1$-blocks, then congruence classes of indices, and finally the indices themselves is
\begin{equation}
p(\eta)\binom{|S|}{m_1}(k-1)^{|S|-m_1}\ektr B^{m_1}  \left( \pfrac{N}{k}^{\eta-|S|} + O\left(\pfrac{N}{k}^{\eta-|S| - 1} \right) \right),
\end{equation}
where the lower order terms in $N$ come from matchings which were proven in Lemma \ref{lem_contributions} to yield fewer degrees of freedom in $N$.
\end{proof}

When computing the $m\text{th}$ moment, the following combinatorial lemma allows us to cancel the contributions of classes with more than $m$ blocks.

\begin{lem}\label{lem Combinatorics Cancellation}
For any $0\le p <m$,
\begin{equation}
\sum_{j=0}^{m} (-1)^j \binom{m}{j} j^p\ =\ 0.
\end{equation}
Furthermore
\begin{equation}
\sum_{j=0}^{m} (-1)^{m-j} \binom{m}{j} j^m\ =\ m!.
\end{equation}
\end{lem}

The proof is  a straightforward calculation; see Appendix \ref{Appendix Combinatorics Proof}.

We are now ready to prove our main result on the moments.

\begin{thm}\label{thm_final_moments}
Denote the centered moments of the empirical blip spectral measure of the $N \times N$ $k$-checkerboard ensemble by $\overline{\mu}_{A,N}^{(m)}$. Then
\begin{equation}
\lim_{N \rightarrow \infty} \E[\overline{\mu}_{A,N}^{(m)}]\ =\ \frac{1}{k} \ektr B^{m}.
\end{equation}
\end{thm}

\begin{proof}
Recall that by Lemma \ref{lem Moments of blip spectral measure},
\begin{equation}\label{eq_above}
\E[\mu_{A,N}^{(m)}]\ =\ \frac{1}{k}\left(\frac{k}{N}\right)^{2n} \sum_{j=0}^{2n} \binom{2n}{j} \sum_{i=0}^{m+j} \binom{m+j}{i} \left(-\frac{N}{k}\right)^{m-i} \etr A^{2n+i}.
\end{equation}

We consider which values of $|S|$ allow a class to actually contribute in the limit. For fixed $j$, by the formula for the expected $m$\textsuperscript{th} moment of $\mu$ given in \eqref{eqn moments of blip spectral measure} and Lemma \ref{lem Combinatorics Cancellation}, the contribution of an $S$-class cancels if $p(\eta)$ has degree less than $m+j$. Hence by the expression for the degree of $p$ given in Proposition \ref{prop_polygon}, an $S$-class cancels if $|S|<m+j$.
However, again by Proposition \ref{prop_polygon}, the contribution of an $S$-class to $\etr A^\eta$ is $O(N^{\eta-|S|})$, which if $|S|>m$ contributes a $o(1)$ term to \eqref{eq_above} after multiplying with the $(k/N)^{2n}$ term. Hence the only contributing $S$-classes have $m+j \leq |S| \leq m$, i.e., $|S|=m$ and $j=0$.

Then we may remove the sum over $j$ to yield
\begin{equation}\label{eq_reduced_moments}
\E[\mu_{A,N}^{(m)}]\ =\ \frac{1}{k}\left(\frac{k}{N}\right)^{2n} \sum_{i=0}^{m} \binom{m}{i} \left(-\frac{N}{k}\right)^{m-i} \etr A^{2n+i}.
\end{equation}
By the previous discussion, Lemma \ref{lem_contributions}, the terms which contribute to $\etr A^{2n+i}$ and do not vanish in the limit arise from classes with $m_1$ $1$-blocks and $(m-m_1)$ $2$-blocks. By Proposition \ref{prop_polygon}, these are of the form
\begin{equation}
p(2n+i)\binom{m}{m_1}(k-1)^{m-m_1}\ektr B^{m_1}  \left(\pfrac{N}{k}^{(2n+i)-m}+O\left(\pfrac{N}{k}^{(2n + i) - m - 1} \right)\right).
\end{equation}

Hence the contribution of such a class to $\E[\mu_{A,N}^{(m)}]$ in the limit is
\begin{align}
&\frac{1}{k}\pfrac{k}{N}^{2n} \sum_{i=0}^{m} \binom{m}{i} \left(-\frac{N}{k}\right)^{m-i}p(2n+i)\binom{m}{m_1}(k-1)^{m-m_1}\ektr B^{m_1} \pfrac{N}{k}^{(2n+i)-m} \nonumber\\
&=\ \frac{1}{k}\binom{m}{m_1}(k-1)^{m-m_1}\ektr B^{m_1} \sum_{i=0}^{m}(-1)^{m-i}\binom{m}{i}p(2n+i).
\end{align}
By the first part of Lemma \ref{lem Combinatorics Cancellation}, all terms in $p(2n+i)$ of degree lower than $m$ in $i$ cancel. Since $p$ is of degree $m$ by Lemma \ref{lem_contributions}, only the highest degree term in $i$ contributes, and this term is equal to ${i^m}/(m!)$ by the same lemma.
Applying the second part of the Lemma \ref{lem Combinatorics Cancellation}, we have that the contribution from our class to the limiting expected $m$\textsuperscript{th} moment is
\begin{equation}
\frac{1}{k}\binom{m}{m_1}(k-1)^{m-m_1}\ektr B^{m_1}\frac{1}{m!}m!\ =\ \frac{1}{k}\binom{m}{m_1}(k-1)^{m-m_1}\ektr B^{m_1}.
\end{equation}

Summing the above contributions over $m_1$, we have that
\begin{equation}\label{eqn uncentered moments}
\lim_{N\to\infty} \E\left[\mu_{A,N} ^{(m)}\right]\ =\ \frac{1}{k}\sum_{m_1 = 0}^{m}  \binom{m}{m_1}(k-1)^{m-m_1} \cdot \E_k \tr B^{m_1}.
\end{equation}

It is natural to compute the centered moments of the distribution.  The uncentered mean is
\begin{equation}
\E\left[\mu_{A,N} ^{(1)}\right]\ =\ k - 1.
\end{equation}

It is not trivial from the definition that centering the limiting expected moments of $\mu_{A,N}$ yields the limiting expected centered moments of $\mu_{A,N}$, but this can be shown straightforwardly from the definitions so we omit the proof. Now, applying the definition of centered moment to the moments given in \eqref{eqn uncentered moments} and reindexing summations gives us that the limiting expected centered moments are
\begin{align}
\mu_c^{(m)}\ &:=\ \lim_{N \rightarrow \infty}\E\left[\int (x-\mu_{A,N}^{(1)})^md\mu_{A,N}\right] \nonumber\\
&=\ \sum_{m_1 = 0}^m \binom{m}{m_1}(-(k-1))^{m-m_1}\E\left[\mu_{A,N} ^{(m_1)}\right] \nonumber\\
&=\ \sum_{m_1 = 0}^m \left[\binom{m}{m_1}(-1)^{m-m_1}(k-1)^{m-m_1}\frac{1}{k}\sum_{i = 0}^{m_1}  \binom{m_1}{i}(k-1)^{m_1-i} \cdot \ektr B^j\right] \nonumber\\
&=\ \sum_{m_1 = 0}^m \left[\binom{m}{m_1}(-1)^{m-m_1} \sum_{i = 0}^{m_1 }  \binom{m_1}{i}(k-1)^{m - i}\frac{1}{k} \ektr B^i\right] \nonumber\\
&=\ \sum_{j=0}^{m} \left[\binom{m}{i}(k-1)^{m-j}\frac{1}{k} \ektr B^i\sum_{m_1 = i}^m \binom{m-j}{m_1 - i}(-1)^{m-m_1} \right].\label{eqn binomial expansion of centered moment}
\end{align}
Now consider the inner sum in \eqref{eqn binomial expansion of centered moment}, which is equal to
\begin{equation}
\sum_{m_1 = 0}^{m-j} \binom{m-j}{m_1} (-1)^{m-m_1}.
\end{equation}
In fact, this is exactly equal to $(-1)^m\delta_{mj}$ where $\delta_{mj}$ is the Kronecker delta function.  From this, the limiting expected centered moments are
\begin{equation}\label{unnormalized centered moments, real case}
\mu_{c}^{(m)}\ =\ \frac{(-1)^m}{k}\ektr B^{m}.
\end{equation}
Because $\ektr B^m=0$ for $m$ odd, we remove the $(-1)^m$ factor, completing the proof.
\end{proof}

Although the formula given by Theorem \ref{thm_final_moments} is implicit, it enables us to compute any $m$\textsuperscript{th} centered moment of the limiting empirical blip spectral measure of the $N \times N$ $k$-checkerboard ensemble by combinatorics on the indices of the corresponding $k\times k$ hollow GOE. We illustrate this with the first few cases.

For ease of notation, we define
\begin{equation}
M_{k,m}\ :=\ \frac{1}{k} \ektr B^{m}.
\end{equation}
For a fixed value of $k$, the $M_{k,m}$'s are the moments of a spectral measure defined upon the $k \times k$ hollow GOE.  Some elementary consequences of this are relevant to our purposes, and we present these below.

\begin{prop}\label{prop_k2gaussian}
For $k = 2$, we have that the $M_{k,m}$'s are the moments of the standard Gaussian.
\end{prop}
\begin{proof}
For $k = 2$, matrices $B$ in the hollow GOE are of the form
\begin{equation}
B\ =\ \begin{bmatrix}
0 & \kmat \\
\kmat & 0
\end{bmatrix}
\end{equation}
for $\kmat \sim \mathcal{N}_{\R}(0, 1)$.

The eigenvalues of $B$ are $\lambda = \pm \kmat$, and the proposition follows immediately.
\end{proof}

\begin{prop} We have
$M_{k,2} = k-1$.
\end{prop}
\begin{proof}
For $B$ a hollow GOE matrix, we have
\begin{equation}
\frac{1}{k} \ektr B^2\ =\ \frac{1}{k} \sum_{1 \leq i,j \leq k} \E[b_{ij}b_{ji}]\ =\ \frac{1}{k}(k^2-k)\ =\ k-1
\end{equation}
upon noting that $\E[b_{ij}b_{ji}]=1$ and $b_{ii}=0$.
\end{proof}


\section{Generalizations to $\C$ and $\H$}\label{sec:genCandH}
We generalize the result of the previous section to complex and quaternion ensembles. Both cases can be reduced to the arguments of the real case in exactly the same manner, so we show only the proof of the quaternion case. The ensembles were defined in Definition \ref{defn:checkerboard}; note we are using  $\complexI, \quaternionJ$  and $\quaternionK$ for the imaginary units to avoid confusion with indices $i,j,k$.

Analogously, we define the hollow GUE and GSE.

\begin{defi}\label{def of zero diagonal Gaussian Unitary Ensemble}
The $k \times k$ \textbf{hollow Gaussian unitary ensemble} and \textbf{hollow Gaussian symplectic ensemble} are the ensembles of matrices $B = (\kmat_{ij})$ given by
\begin{equation} \label{def hollow GUE}
\kmat_{ij}\ =\ \begin{cases}
\mathcal{N}_{\C}(0, 1) \;(\text{resp. }\mathcal{N}_{\H}(0, 1))& \text{{\rm if} } i \neq j \\
0 & \text{{\rm if} } i = j
\end{cases}
\end{equation}
under the restriction $\kmat_{ij}=\overline{\kmat_{ji}}$. We denote the expectation over these ensemble with respect to the natural product probability measures by $\ecktr$ and $\ehktr$.
\end{defi}

We define in addition one new combinatorial notation.

\begin{defi}\label{def congruence configuration}
A \textbf{congruence configuration} is a configuration together with a choice of the congruence class modulo $k$ of every index of a $1$-block.
\end{defi}

The following generalizes Theorem \ref{thm_final_moments} to complex and quaternion ensembles.

\begin{thm}\label{thm quaternion GSE case}
Let $D = \C,\H$, and let $\Eo{D}\left[\mu_{A_N}^{(m)}\right]$ be the expected $m$\textsuperscript{th} moments of the empirical blip spectral measures over the $N \times N$ complex or quaternion $k$-checkerboard ensemble defined as in \eqref{eqn moments of blip spectral measure}. Then
\begin{equation}
\lim_{N \rightarrow \infty} \Eo{D}\left[\mu_{A_N}^{(m)}\right]\ =\ \frac{1}{k} \E_k^{D} \; Tr[B^m]. 
\end{equation}
\end{thm}

\begin{proof}
We prove the quaternion case, and the complex case follows similarly.

To begin, notice that the first statement in Lemma \ref{lem_contributions} regarding counting the arrangements of blocks applies to this case exactly as it was stated.

Now, notice that the $S$-classes for which blocks have size one or two, which were shown to be the only contributing classes in Lemma \ref{lem_contributions}, still have the same number of degrees of freedom in the quaternion case.  It is apparent that a configuration in the quaternion case cannot have more degrees of freedom than the analogous configuration given above in the real case.  It follows that configurations which do not contribute in the real case also do not contribute in the quaternion case.  

Note that because the quaternions are not commutative, we must take care in computing the expectations of the cyclic products.
Consider a configuration. Note that the $w$'s in the configuration are real.  Further, the matched $2$-blocks, i.e., $a_{ij}a_{ji}=|a_{ij}|^2$, are also real. Therefore, all the quarternion-valued $1$-blocks in the configuration commute with the $w$'s and the matched $2$-blocks. Hence, for all cyclic products, we have that the expectation breaks down as
\begin{equation}
\E[\text{Cyclic Product}]\ =\ \E[\text{$1$-blocks (In the order they appear)}] \cdot \E[\text{$2$-blocks and $w$'s}].
\end{equation}

By Lemma \ref{lem_contributions} we need only consider matchings where the $1$-blocks have different indices from the $2$-blocks. Also, given a choice of the congruence classes of the indices of the $1$-blocks, we may construct a corresponding product of the entries in the $k \times k$ hollow GSE given by entries in the $k \times k$ hollow GSE whose indices are those prescribing the congruence class choices on the indices of the $1$-blocks.

Now, suppose that $\Pi_1$ is a congruence configuration of the $k$-checkerboard matrix, and suppose that $\Pi_2$ is the corresponding product of entries of the $k \times k$ hollow GSE. The $a_{ij}$'s that make up these products are quaternions, and hence they do not necessarily commute under multiplication. To deal with this issue, we distribute the product and commute the summed terms to make sure that all the copies of a distinct random variable are placed adjacently in the product. Upon doing this we can use independence of the random variables to convert the expectation of the product into a product of expectations which allows us to compute the expectation using the moments of the Gaussian.\\\\
In particular, we let $a_{ij}=\frac{r_{ij}+\quaternionI x_{ij}+\quaternionJ y_{ij}+\quaternionK z_{ij}}{2}$. Distributing, we get
\begin{equation}
a_{i_1i_2} \cdots a_{i_ni_1}\ =\ \frac{1}{2^n}\sum_{c_{ij}} \prod_{1 \leq \ell \leq n} c_{i_\ell i_{\ell + 1}},
\end{equation}
where the sum over $c_{ij}$ is over all choices of $c_{ij} \in \{r_{ij},\quaternionI x_{ij},\quaternionJ y_{ij},\quaternionK z_{ij}\}$ for each $ij = i_\ell i_{\ell+1}$ and the indices are taken cyclically. Note that this expansion is the same for a product of $1$-blocks in the $k$-checkerboard ensemble and a product of entries in the hollow GSE.

$\Pi_1$ distributes into a product of Gaussian terms times 1, a product of Gaussian terms times $\quaternionI$, a product of Gaussian terms times $\quaternionJ$, and a product of Gaussian terms times $\quaternionK$.  We denote these products by $\Pi_{1}^{\Re}$, $\Pi_{1}^{\quaternionI}$, $\Pi_{1}^{\quaternionJ}$, and $\Pi_{1}^{\quaternionK}$, respectively.  Similarly define $\Pi_{2}^{\Re}$, $\Pi_{2}^{\quaternionI}$, $\Pi_{2}^{\quaternionJ}$, and $\Pi_{2}^{\quaternionK}$. Now, distribution yields
\begin{equation}\label{eqn checkerboard distributive equality 1}
\Pi_1\ =\ \frac{1}{2^n}\left( \Pi_{1}^{\Re} + \Pi_{1}^{\quaternionI}\quaternionI + \Pi_{1}^{\quaternionJ}\quaternionJ + \Pi_{1}^{\quaternionK}\quaternionK\right)
\end{equation}
on the $k$-checkerboard side, and
\begin{equation}\label{eqn checkerboard distributive equality 2}
\Pi_2\ =\ \frac{1}{2^n}\left( \Pi_{2}^{\Re} + \Pi_{2}^{\quaternionI}\quaternionI + \Pi_{2}^{\quaternionJ}\quaternionJ + \Pi_{2}^{\quaternionK}\quaternionK\right)
\end{equation}
on the hollow GSE side.  Note that the $\E[\Pi_{1}^{\quaternionI}]$, $\E[\Pi_{1}^{\quaternionJ}]$, and $\E[\Pi_{1}^{\quaternionK}]$ terms are all zero, because a nonreal coefficient can only occur when there is an unpaired $1$-block $x_{ij},y_{ij}$ or $z_{ij}$. Hence only the real parts $\E[\Pi_1^{\Re}]$ and $\E[\Pi_2^{\Re}]$ remain.

Using \eqref{eqn checkerboard distributive equality 1}, we have
\begin{equation}
\E\left[\sum_{\text{matchings }}\sum_{\text{indexings}} \Pi_1\right]\ =\ \E\left[\sum_{\text{matchings}} \ \ \sum_{r, x, y\ \text{{\rm or} }z \text{ {\rm choices}}}\ \ \sum_{\text{indexings}} \Pi_{1}'^{\Re}\right]
\end{equation}
where $\Pi_1'^{\Re}$ is summed over all matchings, all $4^{length(\Pi_1)/2}$ ways to substitute in either $r,x,y$ or $z$ for each matched pair, and finally all indexings of these products.

Similarly, on the hollow GSE side, we have
\begin{equation}
   \E\left[\Pi_2\right]\ =\ \E\left[\sum_{r,x,y \text{ or }z \text{ choices}}\Pi_{2}'^{\Re}\right] .
\end{equation}
where again the sum is over all ways to substitute an $r,x,y$ or $z$ for the entries in $\Pi'_2$ so that there are an even number of each.
Hence to show equality of the two sides, it suffices to show termwise equality for each summand of $\sum_{r,x,y \text{ or }z \text{ choices}}$. 
Specifically, we must show

\begin{equation}\label{eqn quaternion sums}
\pfrac{k}{N}^{m_{1}} \sum_{\text{\tiny{matchings}}}\sum_{\text{indexings}}\E [\Pi_{1}'^{\Re}]\ =\ \E[ \Pi_{2}'^{\Re}].
\end{equation}

By the same argument as in the GOE case, when there are, say, $2q$ $x$ terms in $\Pi_1'^{\Re}$, there are $(2q-1)!!$ ways to choose a matching of them on the LHS, while on the RHS the expectation of the $2q$ $x$ terms contribute the Gaussian moment $(2q-1)!!$. For any product and choice of matchings there are $\pfrac{N}{k}^{m_{1}}$ ways to choose indexings, cancelling the $\pfrac{k}{N}^{m_{1}}$ on the LHS. Therefore \eqref{eqn quaternion sums} holds, completing the proof of the quaternion case.

The complex case may be proven by the exact same technique of distributing out products of complex Gaussians into products of real Gaussians and arguing as in the real case on these products, proving the theorem.
\end{proof}

\begin{rem}\label{coulda done complex}
It is possible to prove the complex case directly by a more complicated version of the GOE argument, which was the course first taken by the authors before solving the quaternion case. However, the approach outlined in the proof of Theorem \ref{thm quaternion GSE case} is cleaner and more general.
\end{rem}

\section{Almost-sure convergence}\label{sec:convergence}

The traditional way to show weak convergence of empirical spectral measures to a limiting spectral measure (in probability or almost-surely) is to show that the variance (resp. fourth moment) of the $m$\textsuperscript{th} moment, averaged over the $N \times N$ ensemble, is $O(\frac{1}{N})$ (resp. $O\left(\frac{1}{N^2}\right)$). In the case of the blip spectral measure, we encounter a problem: both assertions are false. Heuristically, as $N$ grows, the empirical spectral measures of $N \times N$ matrices from most standard ensembles will all be similar because there is a large and growing number of eigenvalues to average over and so the behavior of individual eigenvalues is drowned out by the average. However, for a $k$-checkerboard matrix there are only $k$ eigenvalues in the blip, so each blip spectral measure is just a collection of $k$ isolated delta spikes distributed randomly according to the limiting spectral computed in Theorem \ref{thm_final_moments}. As such, for fixed $k$ the variance and fourth moment over the ensemble of the general $m$\textsuperscript{th} moment do not go to $0$. We therefore define a modified spectral measure which averages over the eigenvalues of many matrices in order to extend standard techniques.

In order to facilitate the proof of the main convergence result (Theorem \ref{thm_as_convergence}) we first introduce some new notation. In all that follows we fix $k$ and suppress $k$-dependence in our notation for simplicity. Let $\Omega_N$ be the probability space of $N \times N$ $k$-checkerboard matrices with the natural probability measure. Then we define the product probability space
\begin{equation}
\Omega\ :=\ \prod_{N \in \N} \Omega_N.
\end{equation}
By Kolmogorov's extension theorem, this is equipped with a probability measure which agrees with the probability measures on $\Omega_N$ when projected to the $N$\textsuperscript{th} coordinate. Given $\{A_N\}_{N \in \N} \in \Omega$, we denote by $A_N$ the $N \times N$ matrix given by projection to the $N$\textsuperscript{th} coordinate. In what follows, we suppress the subscript $N \in \N$ on elements of $\Omega$, writing them as $\{A_N\}$.

\begin{rem}
\cite{KKMSX} employs a similar construction using product space, while \cite{HM} views elements of $\Omega$ as infinite matrices and the projection map $\Omega \rightarrow \Omega_N$ as simply choosing the upper left $N \times N$ minor.
\end{rem}

Previously we treated the $m$\textsuperscript{th} moment of an empirical spectral measure $\mu_{A,N}^{(m)}$ as a random variable on $\Omega_N$, but we may equivalently treat it as a random variable on $\Omega$. To highlight this, we define the random variable $X_{m,N}$ on $\Omega$
\begin{equation}\label{eq_xmn}
X_{m,N}(\{A_N\})\ :=\ \mu_{A_N,N}^{(m)}.
\end{equation}
These have centered $r$\textsuperscript{th} moment
\begin{equation}
X_{m,N}^{(r)}\ :=\ \E[(X_{m,N}-\E[X_{m,N}])^r].
\end{equation}

Per our motivating discussion at the beginning of this section, because we wish to average over a growing number of matrices of the same size, it is advantageous to work over $\Omega^\N$; this again is equipped with a natural probability measure by Kolmogorov's extension theorem. Its elements are sequences of sequences of matrices, and we denote them by $\overline{A}=\{A^{(i)}\}_{i \in \N}$ where $A^{(i)} \in \Omega$. We now give a more abstract definition of the averaged blip spectral measure defined in Definition \ref{def_ave_blip_fin}.

\begin{defi}\label{def_average_blip_measure}
Fix a function $g: \N \rightarrow \N$. The \textbf{averaged empirical blip spectral measure} associated to $\overline{A} \in \Omega^\N$ is
\begin{equation}
\mu_{N,g,\overline{A}} := \frac{1}{g(N)}\sum_{i=1}^{g(N)} \mu_{A_N^{(i)},N}.
\end{equation}
\end{defi}

In other words, we project onto the $N$\textsuperscript{th} coordinate in each copy of $\Omega$ and then average over the first $g(N)$ of these $N \times N$ matrices.

\begin{rem}\label{rem_stupidblocks}
If one wishes to avoid defining an empirical spectral measure which takes eigenvalues of multiple matrices, one may use the (rather contrived) construction of a $\N \times \N$ block matrix with independent $N \times N$ checkerboard matrix blocks.
\end{rem}

Analogously to $X_{m,N}$, we denote by $Y_{m,N,g}$ the random variable on $\Omega^\N$ defined by the moments of the averaged empirical blip spectral measure
\begin{equation}
Y_{m,N,g}(\overline{A})\ :=\ \mu_{N,g,\overline{A}}^{(m)}.
\end{equation}
The centered $r$\textsuperscript{th} moment (over $\Omega^\N$) of this random variable will be denoted by $Y_{m,N,g}^{(r)}$.

We now prove almost-sure weak convergence of the averaged blip spectral measures under a growth assumption on $g$. Recall the following definition.

\begin{defi}
A sequence of random measures $\{\mu_N\}_{N \in \N}$ on a probability space $\Omega$ converges \textbf{weakly almost-surely} to a fixed measure $\mu$ if, with probability $1$ over $\Omega^\N$, we have
\begin{equation}
\lim_{N \rightarrow \infty} \int f d\mu_N = \int f d\mu
\end{equation}
for all $f \in \mathcal{C}_b(\R)$ (continuous and bounded functions).
\end{defi}

\begin{thm}\label{thm_as_convergence}
Let $g: \N \rightarrow \N$ be such that there exists an $\delta>0$ for which $g(N) = \omega(N^\delta)$. Then, as $N\to\infty$, the averaged empirical spectral measures $\mu_{N,g,\overline{A}}$ of the $k$-checkerboard ensemble converge weakly almost-surely to the measure with moments $M_{k,m}=\frac{1}{k} \ektr[B^m]$, the limiting expected moments computed in Theorem \ref{thm_final_moments}.
\end{thm}
\begin{proof}

For simplicity of notation, we suppress $k$ and denote $M_{k,m}$ by $M_m$. By the triangle inequality, we have
\begin{equation}\label{eqn_tri_ineq_y}
\abs{Y_{m,N,g}-M_m}\ \leq\ \abs{Y_{m,N,g}-\E[Y_{m,N,g}]}+\abs{\E[Y_{m,N,g}]-M_m}.
\end{equation}
From Theorem \ref{thm_final_moments}, we know that $\E[X_{m,N}] \to M_m$, and it follows that $\E[Y_{m,N,g}] \to M_m$. Hence to show that $Y_{m,N,g} \to M_m$ almost surely, it suffices to show that $|Y_{m,N,g} - \E[Y_{m,N,g}]|\to 0$ almost surely as $N \rightarrow \infty$.
We show that the limit as $N \rightarrow \infty$ of all moments over $\Omega_N$ of any arbitrary moment of the empirical spectral measure exists, and that we may always choose a sufficiently high moment\footnote{Note the difference between this and the standard technique of, for instance, \cite{HM}, which uses only the fourth moment.} such that the standard method of Chebyshev's inequality and the Borel-Cantelli lemma gives that $|Y_{m,N,g}-\E[Y_{m,N,g}]|\to 0$. Finally, the moment convergence theorem gives almost-sure weak convergence to the limiting averaged blip spectral measure.

\begin{lem}\label{lem_moments_of_moments}
Let $X_{m,N}$ be as defined in \eqref{eq_xmn}. 
Then for any $t \in \N$, the $r$\textsuperscript{th} centered moment of $X_{m,N}$ satisfies
\begin{equation}
X_{m,N}^{(r)} \ =\ \E\left[\left(X_{m,N}-\E[X_{m,N}]\right)^r\right]\ =\ O_{m,r}(1)
\end{equation}
as $N$ goes to infinity.
\end{lem}
\begin{proof}
After expanding $\E\left[\left(X_{m,N}-\E[X_{m,N}]\right)^r\right]$ binomially, the proof follows similarly to that of Theorem \ref{thm_final_moments}. For more details, see Appendix \ref{app_proof_lem_mom_bound}.
\end{proof}

We apply the following Theorem (Theorem $1.2$ of \cite{Fer}) with $X=X_{m,N}-\E\left[X_{m,N}\right]$, $s=g(N)$ and $\mu_i=X_{m,N}^{(i)}$.

\begin{thm}\label{thm_turkishguy}
Let $r \in \N$ and let $X_1,\ldots,X_s$ be i.i.d. copies of some mean-zero random variable $X$ with absolute moments $\E[|X|^\ell]<\infty$ for all $\ell \in \N$. Then
\begin{equation}
\E\left[ \left( \sum_{i=1}^s X_i\right)^r \right]\ =\ \sum_{1 \leq m \leq \frac{r}{2}}B_{m,r}(\mu_2,\mu_3,\ldots,\mu_r) \binom{s}{m}
\end{equation}
where $\mu_i$ are the moments of $X$ and $B_{m,r}$ is a function independent of $s$, the details of which are given in \cite{Fer}.
\end{thm}
We must first show boundedness of the absolute moments of $X_{m,N}$. By Cauchy-Schwarz,
\begin{equation}
\left(\int |x^{2\ell+1}|d\mu_{X_{m,N}}\right)^2\ \leq\ \int |x|^2 d\mu_{X_{m,N}} \cdot \int |x|^{4\ell}d\mu_{X_{m,N}},
\end{equation}
where $\mu_{X_{m,N}}$ is the probability measure on $\Omega$ given by the density of $X_{m,N}$. Since, for fixed $N$, the even moments of $X_{m,N}$ are finite by \ref{eq_finite_moments}, the previous bound shows that all odd absolute moments are finite as well. Hence Theorem \ref{thm_turkishguy} applies, yielding
\begin{equation}
\E\left[\left(\sum_{i=1}^{g(N)} X_{m,N,i}-\E\left[X_{m,N,i}\right]\right)^r\right]\ =\ \sum_{1 \leq m \leq \frac{r}{2}}B_{m,r}(X_{m,N}^{(2)},X_{m,N}^{(3)},\ldots,X_{m,N}^{(r)}) \binom{g(N)}{m}.
\end{equation}
where the $X_{m,N,i}$ are $i$-indexed i.i.d. copies of $X_{m,N}$. By Lemma \ref{lem_moments_of_moments}, for sufficiently high $N$, $X_{m,N}^{(t)}$ are uniformly bounded above by some constant $K$ for $1 \leq t \leq m$, so there exists $C$ such that $B_{m,r}(X_{m,N}^{(2)},X_{m,N}^{(3)},\ldots,X_{m,N}^{(r)}) < C$ for all sufficiently large $N$ and for all $1\leq m \leq r/2$. Hence
\begin{equation}
\E\left[\left(\sum_{i=1}^{g(N)} X_{m,N,i}-\E\left[X_{m,N,i}\right]\right)^r\right]\ \leq\ \sum_{1\leq m \leq \frac{r}{2}}  C\binom{g(N)}{m}.
\end{equation}
As such, we have
\begin{equation}
Y_{m,N,g}^{(r)} = \frac{1}{g(N)^r} \E\left[\left(\sum_{i=1}^{g(N)} X_{m,N,i}-\E\left[X_{m,N,i}\right]\right)^r\right] \leq \sum_{1\leq m \leq \frac{r}{2}} \frac{C}{g(N)^r}\binom{g(N)}{m}  = O\left(\frac{1}{g(N)^{r/2}}\right).
\end{equation}
Since $g(N) = \omega(N^\delta)$, we may choose $r$ sufficiently large so that
\begin{equation}
Y_{m,N,g}^{(r)}\ =\ O\left(\frac{1}{N^2}\right).
\end{equation}
Then by Chebyshev's inequality,
\begin{equation}
\prob(|Y_{m,N,g}-\E[Y_{m,N,g}]|>\eps)\ \leq\ \frac{\E\left[\left(Y_{m,N,g}-\E[Y_{m,N,g}]\right)^r \right]}{\eps^r}\ =\ \frac{Y_{m,N,g}^{(r)}}{\eps^r}\ =\ O\left(\frac{1}{N^2}\right).
\end{equation}

We now apply the following.

\begin{lem}[Borel-Cantelli]
Let $B_i$ be a sequence of events with $\sum_i \prob(B_i)<\infty$. Then
\begin{equation}
\prob\left(\bigcap_{j=1}^\infty \bigcup_{\ell=j}^\infty B_\ell\right)\ =\ 0.
\end{equation}
\end{lem}

Define the events
\begin{equation}
B_N^{(m,d,g)}\ := \ \left\{A \in \Omega^\N: |Y_{m,N,g}(A) - \E[Y_{m,N,g}]| \ \geq\ \frac{1}{d}\right\}.
\end{equation}
Then $\prob(B_N^{(m,d,g)}) \leq \frac{C_m d^r}{N^2}$, so for fixed $m$, $d$, the conditions of the Borel-Cantelli lemma are satisfied. Hence
\begin{equation}
\prob\left(\bigcap_{j=1}^\infty \bigcup_{\ell=j}^\infty B_\ell^{(m,d,g)}\right) \ =\ 0.
\end{equation}
Taking a union of these measure-zero sets over $d \in \N$ we have
\begin{equation}
\prob\left(Y_{m,N,g} \neq \E[Y_{m,N,g}] \text{ for infinitely many $N$}\right)\ =\ 0,
\end{equation}
and taking the union over $m \in \Z_{\geq 0}$,
\begin{equation}
\prob\left(\exists m \text{ such that }Y_{m,N,g} \neq \E[Y_{m,N,g}] \text{ for infinitely many $N$}\right)\ = \ 0.
\end{equation}
Therefore with probability $1$ over $\Omega^\N$, $|Y_{m,N,g}-\E[Y{m,N,g}]| \to 0$ for each $m$. This, together with \eqref{eqn_tri_ineq_y} and the discussion following it, yields that the moments $\mu_{N,g}^{(m)}=Y_{m,N,g} \to M_m$ almost surely. We now use the following to show almost-sure weak convergence of measures (see for example \cite{Ta}).

\begin{thm}[Moment Convergence Theorem]\label{thm_moment_convergence}
Let $\mu$ be a measure on $\R$ with finite moments $\mu^{(m)}$ for all $m \in \Z_{\geq 0}$, and $\mu_1,\mu_2,\ldots$ a sequence of measures with finite moments $\mu_n^{(m)}$ such that $\lim_{n\rightarrow \infty} \mu_n^{(m)} = \mu^{(m)}$ for all $m \in \Z_{\geq 0}$. If in addition the moments $\mu^{(m)}$ uniquely characterize a measure, then the sequence $\mu_n$ converges weakly to $\mu$.
\end{thm}

To show Carleman's condition is satisfied for the limiting moments $M_m$, we show that $M_m$ are bounded above by the moments of the Gaussian. The odd moments vanish, and by Theorem \ref{thm_final_moments} the even moments are given by
\begin{equation}
M_{2m}\ =\ \frac{1}{k} \ektr A^{2m}\ =\ \sum_{1 \leq i_1,\ldots,i_{2m} \leq k} \E[a_{i_1i_2}a_{i_2i_3}\ldots a_{i_{2m}i_1}],
\end{equation} and as
$\E[a_{i_1i_2}a_{i_2i_3}\ldots a_{i_ni_1}]$ is maximized when all $a_{i_{\ell}i_{\ell+1}}$ are equal,
\begin{equation}
M_{2m}\ \leq\  \sum_{1 \leq i_1,\ldots,i_{2m} \leq k} (2m-1)!!\ =\ k^{2m}(2m-1)!!.
\end{equation}
These are the moments of $\mathcal{N}(0,k)$ so Carleman's condition is satisfied, thus we let $\overline{\mu}$ be the unique measure determined by the moments $M_m$. Choose $\overline{A} \in \Omega^\N$. Then the preceding argument showed that, with probability $1$ over $\overline{A}$ chosen from $\Omega^\N$, all moments $\mu_{N,g,\overline{A}}^{(m)}$ of the measures $\mu_{N,g,\overline{A}}$ converge to $M_m$. Then by Theorem \ref{thm_moment_convergence} the measures $\mu_{N,g,\overline{A}}$ converge weakly to $\overline{\mu}$ with probability $1$, completing the proof.
\end{proof}


\appendix


\section{Details for the Bulk}\label{app_bulk}


In this appendix we give additional details related to \S\ref{sec:bulk}. First, we verify that the expected higher moments of the $(k,1)$-checkerboard ensemble do not converge as $N \rightarrow \infty$. We then demonstrate almost sure weak convergence of the bulk eigenvalues to a semicircle.

\begin{prop} \label{prop:bulk_divergence}
The average moments diverge in the bulk case, namely
\begin{equation}
\E\left[\nu_{A_N}^{(\ell)} \right]\ =\ \Omega(N^{\ell/2-1}).
\end{equation}
\end{prop}

\begin{proof}
By the eigenvalue-trace lemma, we have that
\begin{equation}
\E\left[\nu_{A_N}^{(\ell)} \right]\ = \ \frac{1}{N^{\ell/2+1}}\E\left[\on{Tr}(A_N^\ell) \right]\ =\ \frac{1}{N^{\ell/2+1}} \sum_{1\le i_1,\ldots,i_\ell\le N} \E\left[{a_{i_1i_2} \cdots a_{i_{\ell-1}i_\ell}  a_{i_\ell i_{1}}}\right].
\end{equation}

Note that the expectation of any term in the sum is non-negative. We now count the number of terms where each $a_{ij}=1$. Each such term uniquely corresponds to a choice of $i_1,...,i_\ell$ all congruent to each other modulo $k$. Hence the contribution of these terms is $\Omega(N^\ell)$, which gives the result.
\end{proof}

In \S\ref{sec:bulk}, we established convergence in expectation of the moments. We now show how to extend this to almost sure weak convergence of the empirical densities. This verification is standard, for instance, see \cite{Feier}. To do this, we establish the following lemma.


\begin{lemma}\label{lem:bulk_var_est}
Let $A_N$ be an $N\times N$ $(k,0)$-checkerboard matrix. Then for each fixed $\ell$,
\begin{equation}
\on{Var}(\nu_{A_N}^{(\ell)})\ = \ O(1/N^2).
\end{equation}
\end{lemma}

From this lemma, we can obtain almost sure convergence as follows. Firstly, by Chebyshev's inequality and the previous lemma,

\begin{align}\label{eq:chebyshev_ex}
\sum_{N=1}^\infty {\rm Pr}\pez{\left|\nu_{A_N}^{(\ell)}-\E\left[\nu_{A_N}^{(\ell)}\right]\right|\ >\ \epsilon}&\ \le\  \frac{1}{\epsilon^2}\sum_{N=1}^\infty \on{Var}(\nu_{A_N}^{(\ell)})<\infty.
\end{align}

Hence, by Borel-Cantelli, ${\rm Pr}\pez{\limsup_N  \left|\nu_{A_N}^{(\ell)}-\E\left[\nu_{A_n}^{(\ell)}\right]\right|>\epsilon}=0$, so $\nu_{A_n}^{(\ell)}\rightarrow \E\left[\nu_{A_n}^{(\ell)}\right]$ almost surely, giving us Theorem \ref{thm_main_bulk}  by the method of moments.

\begin{proof}[Proof of Lemma \ref{lem:bulk_var_est}]
This proof is combinatorics. By the eigenvalue trace lemma
\begin{align}
\left|\E\left[(\nu_{A_N}^{(\ell)})^2\right]-\left[\E(\nu_{A_N}^{(\ell)})\right]^2\right|&\ = \ \frac{1}{N^{\ell+2}}\left|\E\left[\on{tr}(A_N^\ell)^2\right]-\pez{\E\left[\on{tr}(A_N^\ell)\right]}^2\right| \nonumber \\
&\ = \ \frac{1}{N^{\ell+2}} \sum_{\textbf{I},\textbf{I}'} \left|\E [\zeta_{\mathbf{I}}\zeta_{\mathbf{I}'}]-\E[\zeta_{\mathbf{I}}]\E[\zeta_{\mathbf{I'}}]\right|, \label{eqn_est_bulk_con}
\end{align}
where $\zeta_{\mathbf{I}}$ is a stand-in for writing out the product $a_{i_1i_2}\cdots a_{i_{\ell-1}i_\ell}a_{i_\ell i_1}$ associated to the sequence $\mathbf{I}=i_1\ldots i_\ell$, where $1\le i_1,\ldots,i_\ell\le N$; hence the sum over pairs $(\mathbf{I},\mathbf{I}')$. Moreover, as in Lemma \ref{lem:avgmomentsbulk}, each pair corresponds to a pair of walks on a graph with vertices $V_{(\mathbf{I},\mathbf{I}')}=\{i_1,\ldots i_\ell,i_1',\ldots i_\ell'\}$ and with edges that we denote as $E_{(\mathbf{I},\mathbf{I}')}$. We say that two such pairs of walks are equivalent if they are equivalent up to relabeling the underlying set of nodes. We then define the weight of $(\mathbf{I},\mathbf{I}')$ to be $\abs{V_{(\mathbf{I},\mathbf{I}')}}$.

We claim that the pairs of weight $t\le \ell$ contribute $O(N^t)$ to the sum. Each equivalence class of weight $t$ gives rise to $O(N^t)$ equivalent pairs as we are choosing $t$ distinct nodes for the labels. Moreover, the contribution of each term is $O(1)$ as the moments of the random entries are finite.

We now consider the entries with weight $t\ge \ell+1$. Note that for the expectation of the term $(\mathbf{I},\mathbf{I}')$ to be nonzero, each edge in $E_{(\mathbf{I},\mathbf{I}')}$ must be traversed twice. In addition, the graphs induced by $\mathbf{I}$ and $\mathbf{I}'$ must share an edge, as otherwise, $\E[\zeta_{\mathbf{I}}\zeta_{\mathbf{I}'}]=\E[\zeta_{\mathbf{I}}]\E[\zeta_{\mathbf{I}'}]$ by independence. Since each edge is traversed twice there are at most $\ell$ unique edges in $E_{(\mathbf{I},\mathbf{I}')}$, which is too few to form a connected graph on $\ell+2$ nodes. Therefore, no pair satisfying the two aforementioned conditions can have weight $\ell+2$. Furthermore, in the case of weight $\ell+1$, there is no such pair either. In this case there are $\ell+1$ nodes and at most $\ell$ unique edges in $E_{(\mathbf{I},\mathbf{I}')}$. Hence as the graph is connected it is a tree. As the walk induced by $\mathbf{I}$ in this graph begins and ends at $i_1$, each edge in the walk is traversed twice: once in each direction. An identical statement holds for the walk induced by $\mathbf{I}'$. Hence as there are exactly two of each edge in $E_{(\mathbf{I},\mathbf{I}')}$ the walks induced by $\mathbf{I}$ and $\mathbf{I'}$ are disjoint, a contradiction. Hence, no pairs of weight greater than $\ell$ contribute to the sum, which, together with \eqref{eqn_est_bulk_con}, gives us Lemma \ref{lem:bulk_var_est}.
\end{proof}

\begin{rem}\label{rem:trace_var}
We note that the previous lemma also establishes that $\var(\tr(A_N^{\ell}))=O(N^{\ell})$.
\end{rem}


\section{Proof of Two Regimes}\label{app_tworegimes}

\noindent \small \emph{This appendix is based on work done by Manuel Fernandez (manuelf@andrew.cmu.edu) and Nicholas Sieger (nsieger@andrew.cmu.edu) at Carnegie Mellon under the supervision of the fifth named author, expanded by the third,  seventh and eighth named authors}.\\ \ \normalsize

In this appendix we demonstrate that checkerboard matrices almost surely have two regimes of eigenvalues, one that is $O(N^{1/2+\epsilon})$ (the bulk) and the other of order $N$ (the blip). To do this, we rely on matrix perturbation theory. In particular, we view a $(k,w)$-checkerboard matrix as the sum of a $(k,0)$-checkerboard matrix and a fixed matrix $Z$ where $Z_{ij}=w\chi_{\{{i \equiv j \bmod k}\}}$. In that sense, we view the $(k,w)$-checkerboard matrix as a perturbation of the matrix $Z$. Then, as the spectral radius of the $(k,0)$-checkerboard matrix is $O(N^{1/2+\epsilon})$, we obtain by standard results in the theory of matrix perturbations that the spectrum of the $(k,w)$-checkerboard matrix is the same as that of matrix $Z$ up to an order $N^{1/2+\epsilon}$ perturbation.

We begin with the following observation on the spectrum of the matrix $Z$:

\begin{lemma}\label{lem_spectrum_Z}
The matrix $Z$ has exactly $k$ non-zero eigenvalues, all of which are equal to $Nw/k$.
\end{lemma}
\begin{proof}
For $1\le j \le k$ the vectors $\sum_{i=0}^{N/k-1} e_{ki+j}$ are eigenvectors with eigenvalues $Nw/k$. Furthermore, for $1\le i\le N/k-1$ and $1\le j<k$  the vector $e_{ki+j}-e_{ki+j+1}$ are eigenvectors with eigenvalues equal to $0$.
\end{proof}

Weyl's inequality gives the following:

\begin{lemma}\label{lem_weyl}
(Weyl's inequality) \cite{HJ} Let $H,P$ be $N\times N$ Hermitian matrices, and let the eigenvalues of $H$, $P$, and $H+P$ be arranged in increasing order. Then for every pair of integers such that $1\le j,\ell\le n$ and $j+\ell\ge n+1$ we have
\begin{equation}\label{eqn_wyle1}
\lambda_{j+\ell-n}(H+P)\le \lambda_j(H)+\lambda_\ell(P),
\end{equation}
and for every pair of integers $j,\ell$ such that $1\le j,\ell\le n$ and $j+\ell\le n+1$ we have
\begin{equation}\label{eqn_wyle2}
\lambda_j(H)+\lambda_\ell(P)\le \lambda_{j+\ell-1}(H+P).
\end{equation}
\end{lemma}

Let $\|P\|_{\text{op}}$ denote $\max_i \left|\lambda_i(P)\right|$. By using the fact that $\left|\lambda_\ell (P)\right|\le \|P\|_{\text{op}}$ and taking $\ell=n$ in \eqref{eqn_wyle1}, we obtain that $\lambda_j(H+P)\le \lambda_j(H)+\|P\|_{\text{op}}$. Taking $\ell=1$ in \eqref{eqn_wyle2} gives the inequality on the other side, hence $\left| \lambda_j(H+P)-\lambda_j(H)\right|\le \|P\|_{\text{op}}$.

The above lemma gives that if the spectral radius of $P$ is $O(f)$ then the size of the perturbations will be $O(f)$ as well. Hence it suffices to demonstrate that almost surely the spectral radius of a sequence of $(k,0)$-checkerboard matrices is $O(N^{1/2+\epsilon})$.


\begin{lemma}\label{lem_radius}
Let $m \in \N$ and let $\{A_N\}_{N\in\N}$ be a sequence of $(k,0)$-checkerboard matrices, then almost surely, as $N \to \infty$, $\|A_N\|_{\text{op}}=O(N^{1/2+1/(2m)})$.
\end{lemma}

\begin{proof}
Suppose that for some $m \in \N$ we have a sequence $\{A_N\}_{N \in \N}$ such that $\|A_N\|_{\text{op}}$ is not $O(N^{1/2+1/(2m)})$. Then, by straightforward calculation the $(2m+2)^{nd}$ moments $\mu_{A_N,N}^{(2m+2)}$ do not converge. Hence if $\Pr\{\|A_N\|_{\text{op}}=O(N^{1/2+1/(2m)})\} \neq 1$, then with nonzero probability $\mu_{A_N,N}^{(2m+2)}$ do not converge. This contradicts the almost-sure moment convergence result of Appendix \ref{app_bulk}, and Lemma \ref{lem_radius} follows.
\end{proof}

Since Lemma \ref{lem_radius} holds for all $m\in \N$, we have that almost surely $\|A_N\|_{\text{op}}$ is $O(N^{{1/2}+\epsilon})$. Together with Lemma \ref{lem_spectrum_Z} and Lemma \ref{lem_weyl}, we obtain

\begin{theorem}
Let $\{A_N\}_{N\in\N}$ be a sequence of $(k,w)$-checkerboard matrices. Then almost surely as $N\rightarrow \infty$ the eigenvalues of $A_N$ fall into two regimes: $N-k$ of the eigenvalues are $O(N^{1/2+\epsilon})$ and $k$ eigenvalues are of magnitude $Nw/k+O(N^{{1/2}+\epsilon})$.
\end{theorem}

\section{Proof of Lemma \ref{lem Combinatorics Cancellation}}\label{Appendix Combinatorics Proof}
In \S\ref{sec The Moments of the Blip Spectral Measure}, we introduced Lemma \ref{lem Combinatorics Cancellation} without proof.  Here we provide a short proof of it.

\begin{proof}
Consider the function
\begin{equation}
f_0(x) \ = \  (1-x)^{m} \ = \  \sum_{j=0}^{m} (-1)^j \binom{m}{j} x^j.
\end{equation}
We inductively define, for each $0 \le p < m-1$, the function $f_{p+1}(x) = x f_{p}'(x)$. One can prove by straightforward induction that
\begin{equation}
f_p(x) \ = \  \sum_{i=1}^p c_{i,p} x^i(1-x)^{m-i},
\end{equation}
for each $0\le p < m$, with $c_{i,p}\in\R$, by using the product rule. Therefore, for each $0\le p < m$
\begin{equation}
0 \ = \  f_p(1) \ = \  \sum_{j=0}^{m} (-1)^j \binom{m}{j}j^p.
\end{equation}
By the same reasoning,
\begin{equation}
\sum_{j=0}^{m} (-1)^j \binom{m}{j}j^m \ = \  f_m(1) \ = \  (-1)^m m!
\end{equation}
and the second claim follows.
\end{proof}

\section{Bounds for $X_{m,N}^{(r)}$}
\begin{proof}[Proof of Lemma \ref{lem_moments_of_moments}]\label{app_proof_lem_mom_bound}
Firstly, we have
\begin{align}
 \E\left[\left(X_{m,N}-\E[X_{m,N}]\right)^r\right]\ &\ =\  \E\left[\sum_{\ell=0}^r \binom{r}{\ell}(X_{m,N})^\ell \left(\E[X_{m,N}]\right)^{r-\ell}\right]\  \nonumber \\
 &\ =\ \ \sum_{\ell=0}^r \binom{r}{\ell}\E\left[(X_{m,N})^\ell\right] \left(\E[X_{m,N}]\right)^{r-\ell}
\end{align}
By \eqref{eqn uncentered moments}, we have $\E[X_{m,N}]=O_m(1)$ hence $\left(\E[X_{m,N}]\right)^{r-\ell}=O_{m,r,\ell}(1)$ for all $\ell$. As such, it suffices to show that $\E\left[(X_{m,N})^\ell\right]=O_{m,\ell}(1)$.
By \eqref{eqn moments of blip spectral measure}, we have that
\begin{align}
&\E[X_{m,N}^\ell] \\
&\ =\ \left(\frac{k}{N}\right)^{2n\ell} \E\left[ \left(\sum_{j=0}^{2n} \binom{2n}{j} \sum_{i=0}^{m+j} \binom{m+j}{i} \left(-\frac{N}{k}\right)^{m-i} \tr A^{2n+i}\right)^\ell\right]  \nonumber \\
&\ =\  \left(\frac{k}{N}\right)^{2n\ell} \E\left[\sum_{j_1=0}^{2n} \dots \sum_{j_\ell=0}^{2n} \left[\prod_{u=1}^\ell \binom{2n}{j_u} \right] \sum_{i_1=0}^{m+j_1}\dots \sum_{i_\ell=0}^{m+j_\ell} \left[ \prod_{v=1}^{m+j_v} \binom{m+j_v}{i_v} \right] \left(-\frac{N}{k}\right)^{m-i_v}\tr A^{2n+i_v} \right]  \nonumber \\
&\ =\  \left(\frac{k}{N}\right)^{2n\ell} \sum_{j_1=0}^{2n} \dots \sum_{j_\ell=0}^{2n} \left[\prod_{u=1}^\ell \binom{2n}{j_u}\right] \sum_{i_1=0}^{m+j_1}\dots \sum_{i_\ell=0}^{m+j_\ell} \left[\prod_{v=1}^\ell \binom{m+j_v}{i_v}\right] \left(-\frac{N}{k}\right)^{m-i_v}\E\left[\prod_{v=1}^\ell\tr A^{2n+i_v} \right]. \label{eq_finite_moments}
\end{align}
Now, recall that
\begin{equation}
\E\left[\prod_{v=1}^\ell\tr A^{2n+i_v} \right] \ = \  \sum_{\alpha^1_1,\dots,\alpha^1_{2n+i_1} \leq N} \dots \sum_{\alpha^\ell_1,\dots,\alpha^\ell_{2n+i_\ell}\leq N} \E\left [ \prod_{j=1}^\ell a_{\alpha^j_1,\alpha^j_2}\dots a_{\alpha^j_{2n+i_j},\alpha^j_1} \right].
\end{equation}
We have now reached a combinatorial problem similar to the one we encounter in \S\ref{sec The Moments of the Blip Spectral Measure}. For each $j$, since the length of the cyclic product $a_{\alpha^j_1,\alpha^j_2}\dots a_{\alpha^j_{2n+i_j},\alpha^j_1}$ is fixed at $2n+i_j$, we can choose the number of blocks (determining the class), the location of the blocks (determining the configuration), the matchings and indexings. By Lemma \ref{lem_contributions}, we have that the main contribution from configurations of length $(2n+i_j)$ in $B_j$-class is $\frac{(2n+i_j)^{B_j}}{B_j!}$. By the same arguments made in \S\ref{sec The Moments of the Blip Spectral Measure}, the number of ways we can choose the number of blocks having one $a$ and two $a$'s as well as the number of ways to choose matchings across the $\ell$ cyclic products are independent of $N$, $j$'s and $i_j$'s, so for simplicity, we are denoting them as $C$. Finally, the contribution from choosing the indices of all the blocks and $w$'s is $O_k(N^{2n\ell+i_1+\dots+i_\ell-B_1-\dots-B_\ell})$. As such, if $B_1,\dots,B_\ell \geq m$, the total contribution is $O_{m,k}(1)$.
If there exists $B_{j'}<m$, then the overall contribution is
\begin{equation}
CN^{\ell m-B_1-\dots-B_\ell} \prod_{u=1}^\ell \left[\sum_{j_u=0}^{2n} \binom{2n}{j_u} \sum_{i_u=0}^{m+j_u} \binom{m+j_u}{i_u}(-1)^{m-i_u}\frac{(2n+i_u)^{B_u}}{B_u!}\right]\ = \ 0.
\end{equation}
since the sum over $j_u=j'$ is equal to 0 by by Lemma \ref{lem Combinatorics Cancellation}. As such, the total contribution of $\E[X_{m,N}^\ell]$ is simply $O_{m,\ell}(1)$ (suppressing $k$), as desired.
\end{proof}


\section{Moment Convergence Theorem}\label{sec:momconvthm}


The following argument is standard (though usually assumes all measures concerned are probability measures), and is given for completeness.

\begin{defi}\label{def_tight}
A sequence of measures $(\mu_n)_{n \geq 1}$ on $\R$ is \textbf{uniformly tight} if, for every $\eps>0$, there is a compact set $K$ such that $\sup_{n \geq 1} \mu_n(\R \setminus K)<\infty$.
\end{defi}

We are now ready to prove the moment convergence theorem for general finite measures, largely following the treatment of \cite{Cha}.

\begin{proof}[Proof of Theorem \ref{thm_moment_convergence}]
By convergence of moments, we have that
\begin{equation}
C_k\ :=\ \sup_{n \geq 1} \int_\R x^kd\mu_n
\end{equation}
is bounded. For any $R>0$ we then have by Chebyshev's inequality that
\begin{equation}\label{eqn_unif_tight}
\mu_n(\R \setminus [-R,R])\ \leq\ \frac{\int_\R x^2 d\mu_n}{R^2}\ \leq\ \frac{C_2}{R^2}.
\end{equation}
Therefore the $\mu_n$ are uniformly tight. Hence by Prokhorov's theorem for general measures (see \cite{Bog}, Theorem $8.6.2$), every subsequence of $(\mu_n)_{n\geq 1}$ contains a weakly convergent subsequence which converges to some measure $\nu$.

For any subsequence $(\mu_{n_\ell})_{\ell \geq 1}$ converging weakly to some measure $\nu$, we show that $\nu = \mu$. Fix some $k \in \Z_{\geq 0}$ and $R \in \R_{>0}$. Let $\varphi_R$ be a continuous function such that
\begin{equation}
1_{[-R,R]}\ \leq\ \varphi_R \ \leq\ 1_{[-R-1,R+1]}.
\end{equation}
We may split the integral as
\begin{equation}
\int x^k  d\mu_{n_\ell}\ =\ \int x^k \varphi_R d\mu_{n_\ell} + \int x^k (1-\varphi_R) d\mu_{n_\ell}.
\end{equation}
By the Cauchy-Schwarz inequality,
\begin{equation}
\left| \int x^k(1-\varphi_R)d\mu_{n_\ell} \right|^2 \leq \int x^{2k} d\mu_{n_\ell} \cdot \int (1-\varphi_R)^2 d\mu_{n_\ell}\ \leq\ \int x^{2k} d\mu_{n_\ell} \cdot \mu_{n_\ell}(\R \setminus [-R,R]),
\end{equation}
and by our moment bounds and the definition of $\varphi_R$, this is $\leq \frac{C_2 \cdot C_{2k}}{R^2}$. Therefore, we have
\begin{equation}
\lim_{R \rightarrow \infty} \int x^k \varphi_R d\mu_{n_\ell}\ =\ \int x^k d\mu_{n_\ell}.
\end{equation}
By the moment convergence assumption
\begin{equation}
\lim_{\ell \rightarrow \infty} \int x^k d\mu_{n_\ell}\ =\ \int x^k d\mu
\end{equation}
and by weak convergence,
\begin{equation}
\lim_{\ell \rightarrow \infty} \int x^k \varphi_R d\mu_{n_\ell} = \int x^k \varphi_R d\nu .
\end{equation}
We must now show
\begin{equation}
    \lim_{\ell \rightarrow \infty} \lim_{R \rightarrow \infty}  \int x^k \varphi_R d\mu_{n_\ell} =   \lim_{R \rightarrow \infty} \lim_{\ell \rightarrow \infty} \int x^k \varphi_R d\mu_{n_\ell}.
\end{equation}
For this, it suffices to show that $\int x^k \varphi_R d\mu_{n_\ell}$ converges uniformly with respect to $R$ as $\ell \rightarrow \infty$. In the following argument, we assume $k$ is even so that $x^k$ is nonnegative, but this may be modified easily for $k$ odd. By the same argument for uniform tightness as in Equation \eqref{eqn_unif_tight}, there exists $C$ such that for all $\ell$ (and when replacing $\mu_{n_\ell}$ by $\nu$),
\begin{equation}
\int_{\R \setminus [-K,K]} x^k \varphi_R d\mu_{n_\ell} \ \leq \ \int_{\R \setminus [-K,K]} x^k d\mu_{n_\ell} \ \leq\  \frac{C}{K^2}.
\end{equation}
Hence for any $\eps>0$, there exists some $K$ such that
\begin{equation}
\int_{\R \setminus [-K,K]} x^k \varphi_R d\mu_{n_\ell} < \eps/3
\end{equation}
unconditionally on $\ell$. By weak convergence, for any fixed $R$, there also exists an $N_R$ so that
\begin{equation}
\left|\int x^k \varphi_R d\mu_{n_\ell} - \int x^k \varphi_R d\nu\right| < \eps/3
\end{equation}
for all $\ell > N_R$. Therefore, letting $\overline{N} = \sup_{\substack{R \in \N \\ R \leq K}} N_R$,
we have that for $\ell > \overline{N}$ and any $R$,
\begin{align}
&\left|\int x^k \varphi_R d\mu_{n_\ell} - \int x^k \varphi_R d\nu\right| \\
& \leq \left|\int_{[-K,K]} x^k \varphi_R d\mu_{n_\ell} - \int_{[-K,K]} x^k \varphi_R d\nu\right| + \left|\int_{\R \setminus [-K,K]} x^k \varphi_R d\mu_{n_\ell} \right| + \left|\int_{\R \setminus [-K,K]} x^k \varphi_R d\nu \right| \\
& < \left|\int_{\R} x^k \varphi_{(K-1)} d\mu_{n_\ell} - \int_\R x^k \varphi_{(K-1)} d\nu\right| + \frac{2}{3}\eps \\
& < \eps
\end{align}
Thus we have uniform convergence, so the limits may be switched. Putting all this together,
\begin{align}
\lim_{R \rightarrow \infty} \int x^k \varphi_R d\nu &= \lim_{R \rightarrow \infty} \lim_{\ell \rightarrow \infty} \int x^k \varphi_R d\mu_{n_\ell} \\
&=   \lim_{\ell \rightarrow \infty} \lim_{R \rightarrow \infty} \int x^k \varphi_R d\mu_{n_\ell} \\
&= \lim_{\ell \rightarrow \infty} \int x^k d\mu_{n_\ell} \\
& = \int x^k d\mu
\end{align}
with the last equality following by the moment convergence hypothesis.
We have $\varphi_R x^{2k} \leq \varphi_{R+1}x^{2k}$ and both are nonnegative, so by the monotone convergence theorem
\begin{equation}
\int x^{2k} d\mu\ =\ \lim_{R \rightarrow \infty} \int \varphi_R x^{2k} d\nu\ =\ \int x^{2k}d\nu.
\end{equation}
Hence $x^{k} \in L^2(\nu)$, so $x^{k} \in L^1(\nu)$. Since $\varphi_R x^k \leq x^k$, by the dominated convergence theorem 
\begin{equation}
\int x^k d\mu\ =\ \lim_{R \rightarrow \infty} \int \varphi_R x^k d\nu\ =\ \int x^k d\nu.
\end{equation}
Since $\mu$ is uniquely characterized by its moments, $\nu=\mu$. Since every subsequence of $(\mu_n)_{n \geq 1}$ has a subsequence weakly converging to $\mu$, standard arguments give that $\mu_n$ converges weakly to $\mu$.
\end{proof}



\ \\

\end{document}